\documentclass[
 amsmath,amssymb,
 twocolumn,
prx,
]{revtex4-2}
\usepackage[T1]{fontenc}
\usepackage{graphicx}
\usepackage{dcolumn}
\usepackage{bm}
\usepackage{hyperref}
\usepackage{amsmath,amssymb,amsthm}
\theoremstyle{definition}
\newtheorem{definition}{Definition}[section]
\newtheorem{theorem}{Theorem}[section]
\newtheorem{corollary}{Corollary}[theorem]

\newtheorem{lemma}[theorem]{Lemma}
\theoremstyle{remark}

\usepackage{caption}
\usepackage{physics}
\usepackage{sublabel}
\usepackage{subcaption}
\usepackage{xcolor}
\hypersetup{colorlinks = true, linkcolor=blue, citecolor=blue, urlcolor=blue}

\begin{document}

\preprint{APS/123-QED}

\title{Achievable rates for concatenated square Gottesman-Kitaev-Preskill codes}

\author{Mahadevan Subramanian}
\author{Guo Zheng}%
\author{Liang Jiang}
\affiliation{%
Pritzker School of Molecular Engineering, The University of Chicago, Chicago, Illinois 60637, USA 
}%

\date{\today}

\begin{abstract}
The Gottesman-Kitaev-Preskill (GKP) codes are known to achieve optimal rates under displacement noise and pure loss channels, which establishes theoretical foundations for its optimality. However, such optimal rates are only known to be achieved at a discrete set of noise strength with the current self-dual symplectic lattice construction. In this work, we develop a new coding strategy using concatenated continuous variable -- discrete variable encodings to go beyond past results and establish GKP's optimal rate over all noise strengths. In particular, for displacement noise, the rate is obtained through a constructive approach by concatenating GKP codes with a quantum polar code and analog decoding. For pure loss channel, we prove the existence of capacity-achieving GKP codes through a random coding approach. These results highlight the capability of concatenation-based GKP codes and provides new methods for constructing good GKP lattices. 
\end{abstract}

\maketitle

\section{Introduction}
The goal of reliably communicating quantum information through noisy channels lies at the heart of proposals for quantum networks \cite{Simon2017,5485004,10.1007/978-3-540-70918-3_52,gandotra2025quantumcommunicationbandwidthandtimelimitedchannels,Kimble2008} as well as quantum transduction \cite{Wang2022}. A key aspect of this is understanding what rates of information transmission can be considered achievable \cite{PhysRevA.57.4153}, which has laid the foundation of the study of quantum capacities of noisy channels \cite{PhysRevA.55.1613,PhysRevA.54.2629}. Showing a rate to be achievable offers a direct lower bound to the capacity of the noise channel. However, explicitly showing achievability of rates requires proving existence of a sequence of encoding and decoding schemes that are asymptotically reliable for that particular rate. Due to the non-trivial nature of this task, the coherent information of a channel is often used to find achievable rates. The coherent information is the quantum analogue to mutual information and offers a lower bound to the quantum capacity \cite{PhysRevA.54.2629}. For degradable channels such as the bosonic pure-loss and amplification channels, the maximal one-shot coherent information equals the capacity. Upper and lower bounds on the capacity for channels such as the Gaussian displacement noise \cite{holevo1999evaluatingcapacitiesbosonicgaussian,PhysRevA.64.062301}, loss-dephasing channel \cite{Leviant2022quantumcapacity} and thermal loss channels \cite{8482307} have been found while the exact capacity remains unknown. 

Gottesman-Kitaev-Preskill (GKP) codes \cite{gottesman_encoding_2001} are bosonic error correcting codes that encode a finite dimensional Hilbert space into bosonic modes. These codes are obtained by using a stabilizer group generated by optical phase space displacements, hence are deeply connected to symplectically integral lattices \cite{PRXQuantum.3.010335,Conrad2022gottesmankitaev}. This connection has been used to show that GKP codes can achieve the coherent information of the Gaussian displacement noise channel \cite{PhysRevA.64.062301} as well as the capacity of the pure-loss and amplification channels \cite{zheng2024performanceachievableratesgottesmankitaevpreskill}. However, these results have two major limitations. The rate $R$ (in qubit per mode) these codes achieve are restricted to have $2^R$ be an integer, and these are existence based results which do not provide constructive ways to obtain rate achieving codes. As a result, the optimal rates are only achieved for a discrete set of noise strengths. All known asymptotic GKP coding rate results \cite{PhysRevA.64.062301,zheng2024performanceachievableratesgottesmankitaevpreskill} make use of the family of lattices obtained by rescaling self-dual (unimodular) lattices which as a corollary of the Buser-Sarnak theorem \cite{buser1994period}, will always contain a good spherical packing lattice. This good spherical packing can be used to prove achievability for certain rates. However, to encode information self-dual lattices must be scaled by a square root of an integer (ensuring that it remains valid for encoding a GKP code), which forces the number of logical dimensions per mode to be an integer. Searching for the lattices that satisfy these properties is a hard task \cite{harrington2004analysis} and the current knowledge of optimal sphere-packing lattices is limited to only dimensions of 24 \cite{cohn2022universal}. While the capacity achieving lattices need not be truly optimal in regard to sphere-packing, this does highlight the difficulty one faces in trying to construct such lattices.

In this work, we develop a different approach to overcome this limitation, enabling the achievement of favorable rates through more explicit construction methods. We focus on GKP codes that are obtained by taking $N$ single mode square GKP codes, each encoding a qudit of $d$ levels and then defining an outer code for concatenating as a $N$ qudit stabilizer code encoding $K$ qudits which we denote by $[[N,K]]_d$. For the Gaussian displacement noise channel, we first find that the effective noise channel for a single-mode GKP square qudit is a classical mixture of Pauli noise where knowledge of the syndrome gives us knowledge of what the Pauli noise looks like. We formalize this by finding an achievable rate for this noise channel with the use of quantum polar codes \cite{PhysRevLett.109.050504} that take advantage of the analog information offered by this syndrome. The usefulness of GKP analog information has been well studied \cite{Raveendran2022finiterateqldpcgkp,PhysRevLett.119.180507}, to which we show a rigorous way to quantify the advantage analog information offers. We find that this newly defined achievable rate approaches the coherent information of the Gaussian displacement noise channel at all values of noise strength $\sigma$. Crucially, these GKP codes can be explicitly constructed, and also have efficient ($\mathcal{O}(N\log N)$ circuit depth) encoding and decoding procedures due to the outer concatenation with the quantum polar code. We further support this with a numerical study of the performance of these polar codes for increasing numbers of modes under the Gaussian displacement noise model.

We also show that concatenated GKP codes achieve the capacity of pure-loss, and hence by extension also the capacity of the amplification channel \cite{zheng2024performanceachievableratesgottesmankitaevpreskill}. We use a construction of stabilizer codes for prime $d$ dimensional qudits mapped from self-orthogonal codes in the Galois field $GF(d^2)$ from \cite{959288}. By averaging over the set of GKP codes obtained in this manner, we prove existence of a capacity achieving sequence of codes as $d\to\infty$, in essence recreating the behavior of a good spherical packing lattice. These results highlight the importance of the family of codes obtained by simply concatenating qudit codes to square GKP codes, as well as their capability in exceeding previously known achievable rates for GKP codes. 

This paper is structured as follows. In section \ref{sec:squareGKP}, we introduce the square GKP code, following which we provide a short introduction to achievable rates for communication of quantum information in section \ref{sec:achievable}. Our results begin with section \ref{sec:logicalNoisesq} where show how the effective logical noise channel of a corrected square GKP qudit is a classical mixture of Pauli noise with an auxiliary output containing syndrome information. We then show in section \ref{sec:Isqanalog} how the coherent information of this logical noise channel which we refer to as $I^{\mathrm{sq}}_{d,\mathrm{analog}}$ approaches the coherent information of the Gaussian displacement noise channel for all values of noise strength $\sigma$ for large enough $d$. In section \ref{sec:classical_polar} we provide a brief introduction to polar codes and highlight the construction of quantum polar codes we make use of in section \ref{sec:css_polar} and provide our numeric study of these codes in the context of achieving $I^{\mathrm{sq}}_{d,\mathrm{analog}}$ in section \ref{sec:numeric_polar}. In section \ref{sec:caplossamp}  we discuss the capacity of the bosonic pure-loss and amplification channels following which we show the existence of a capacity achieving sequence for the pure-loss channel by concatenation with square GKP codes in section \ref{sec:lossresult}. Finally we discuss the relevance and implications of our results in section \ref{sec:discuss}.

\section{Background \& notation}
\subsection{The square GKP code}\label{sec:squareGKP}
We begin our discussion focusing on the square GKP code in one mode. We make use of canonical position and momentum operators $\hat{q}$ and $\hat{p}$ respectively which satisfy the commutation relation $[\hat{q},\hat{p}] = i$. The stabilizers are chosen as phase space displacements of equal length in the direction of $\hat{q}$ and $\hat{p}$. We consider the two stabilizers
\begin{equation}
    \hat{S}_1 = \exp(i\hat{q}\sqrt{2\pi d}),\quad\hat{S}_2 = \exp(-i\hat{p}\sqrt{2\pi d}),
\end{equation}
where $d$ is some prime number. For this stabilizer set, the displacements which commute with both $\hat{S}_1$ and $\hat{S}_2$ can be written as
\begin{equation}
    \hat{P}_{u,v} = \exp(i(-u\hat{p}+v\hat{q})\sqrt{\frac{2\pi}{d}}),
\end{equation}
where $u,v\in\mathbb{Z}$. This then gives logical operators $\hat{X}_L = \hat{P}_{1,0}$ and $\hat{Z}_L = \hat{P}_{0,1}$ which can be used to construct the Pauli group. To see how, we can note that $\hat{X}_L\hat{Z}_L = e^{-2\pi/d}\hat{Z}_L\hat{X}_L$ giving the necessary commutation relation for the qudit Pauli group \cite{Sarkar2024quditpauligroupnon}. The eigenstates of $\hat{Z}_L$ and $\hat{X}_L$ would then be
\begin{equation}
\begin{aligned}
    \ket{j_L} = \sum_{n=-\infty}^{\infty}\ket{q=(dn+j)\sqrt{2\pi/d}},\\ \ket{\tilde{j}_L} = \sum_{n=-\infty}^{\infty}\ket{p=(dn+j)\sqrt{2\pi/d}},
\end{aligned}\label{eq:codestates}
\end{equation}
which would satisfy
\begin{equation}
    \hat{X}^u\ket{j_L} = \ket{(j\oplus u)_L},\quad \hat{Z}^v\ket{j_L} = e^{\frac{2\pi i}{d}vj}\ket{j_L}.
\end{equation}
To construct the Pauli group out of $\hat{P}_{u,v}$, we restrict $u,v\in \mathbb{F}_d$ which is the finite-field associated to the prime number $d$. We will be restricting our analysis to prime $d$ since it allows us to ensure that we can define the Galois-Field $GF(d)$ \footnote{For $d$ being a positive integer power of a prime number, a Galois field $GF(d)$ is a finite field with exactly $d$ elements which supports addition, subtraction, multiplication and division (except by $0$). For a quick intro to Galois-Fields, we refer the reader to \cite{blahut2003algebraic}}. Hence $\hat{P}_{u,v}$ can be understood as the logical operation $\hat{X}^u\hat{Z}^v$ (up to a phase) and hence generates the qudit Pauli group $\mathcal{P}_d$.
\begin{figure*}[ht]
    \centering
    \includegraphics[width=\linewidth]{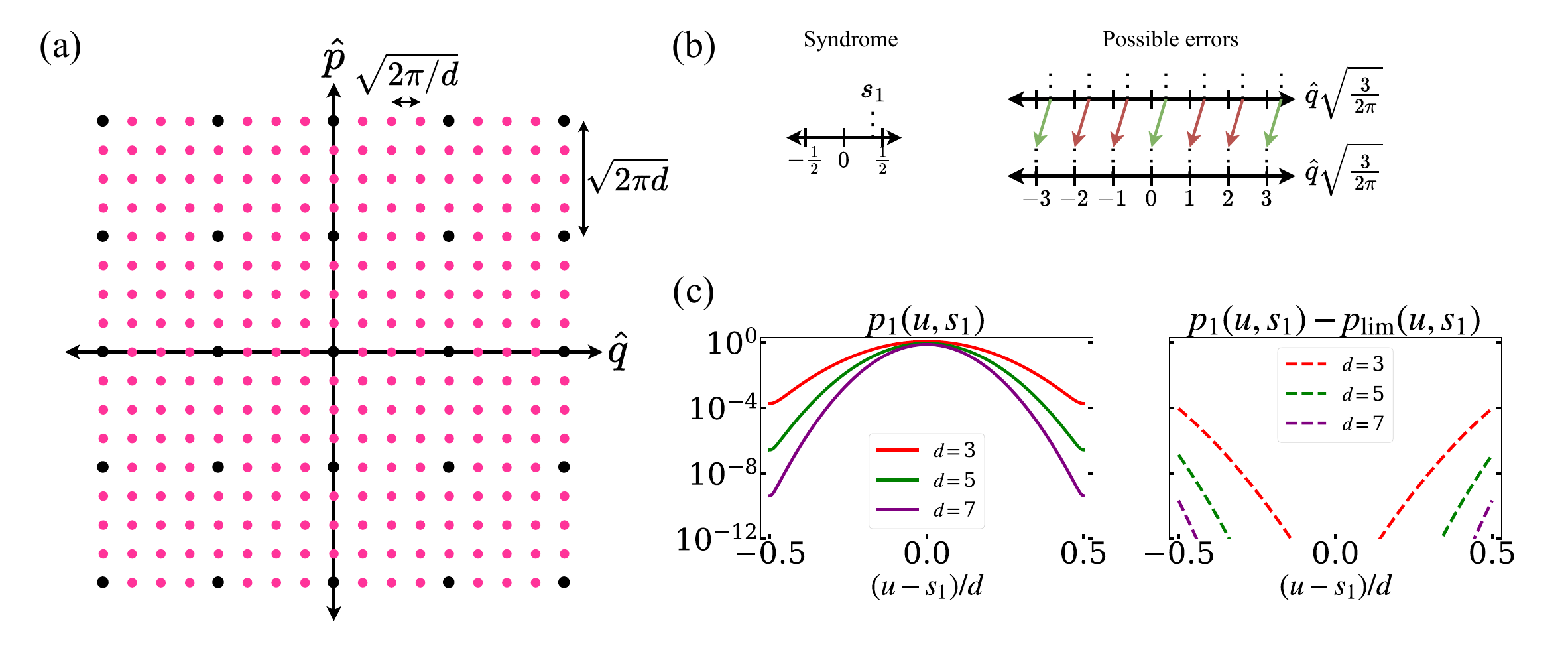}
    \caption{(a) Depiction of the GKP lattice for a square GKP qudit of $d$ levels. Displacements in the stabilizer group lie in the square lattice spaced by $\sqrt{2\pi d}$ (black dots) and the displacements giving logical operations are in the square lattice spaced apart by $\sqrt{2\pi/d}$ (pink dots). (b) Figure depicting relation between the obtained syndrome and logical errors for the case of $d=3$. Obtaining a particular syndrome $s_1$ implies a possible set of errors (see top right line graph) for the shift along $\hat{q}\sqrt{\frac{3}{2\pi}}$. These on correction by the smallest possible displacement for getting back to the logical codespace may result in logical errors (red arrows) or result in no logical error at all (green arrows). (c) By noting that the value of $p_1(u,s_1)$ is only a function of $u+s_1$, we compare these distributions for different values of $d$ for the value of $\sigma=0.5$. In our convention, $(u+s_1)/d$ will always lie in $[-1/2,1/2)$ and we can see that over this normalized range, there is a strong concentration near $0$ for $p_1(u,s_1)$ as $d$ is increased. Further, we plot the difference $p_1(u,s_1)-p_{\mathrm{lim}}(u,s_1)$ (which is also purely a function of $(u+s_1)/d$) and note that increasing the value of $d$ shows a clear exponential suppression in this difference.}
    \label{fig:1}
\end{figure*}
Let the state $\ket{\psi}$ stabilized by both $\hat{S}_1$ and $\hat{S}_2$. We now consider an erroneous displacement to act on this state given by
\begin{equation}
    \hat{E} = \exp(i(-e_1\hat{p}+e_2\hat{q})\sqrt{\frac{2\pi}{d}}),
\end{equation}
giving the displaced state $\ket{\psi_e} = \hat{E}\ket{\psi}$. It then follows that
\begin{equation}
    \hat{S}_1\ket{\psi_e} = \exp(2\pi i e_1)\ket{\psi_e},\quad\hat{S}_2\ket{\psi_e}=\exp(-2\pi i e_2)\ket{\psi_e},
\end{equation}
which means that measuring the eigenvalues of the stabilizer is equivalent to measuring the two values
\begin{equation}
    s_1 = \left(\hat{q}\sqrt{\frac{d}{2\pi}}\right)\mod1,\quad s_2 = \left(-\hat{p}\sqrt{\frac{d}{2\pi}}\right)\mod1,
\end{equation}
since the eigenvalues lie on a rotor for displacement operations. Importantly, each pair of values $s_1,s_2$ define an orthogonal subspace of the full Hilbert space (also referred to as the Zak basis \cite{PhysRevA.107.062611,PRXQuantum.5.010331}). There are broadly two ways of measuring the stabilizers for correction, namely the Steane type measurement and the Knill (teleportation) scheme. The first is the Steane-inspired scheme introduced in \cite{PhysRevA.64.012310} which makes use of CSUM gates and one GKP ancilla with Homodyne measurement. The second is the teleportation based scheme which uses an encoded bell pair \cite{PhysRevA.102.062411,PhysRevA.108.052413} and a beam-splitter and squeezing operations. The teleportation-based scheme propagates errors with a smaller pre-factor compared to the Steane-type scheme which leads to better performance assuming equally noisy ancillas for both schemes \cite{PhysRevA.101.012316,PRXQuantum.3.010315}.

The code states defined in Eq. \eqref{eq:codestates} are infinite energy states since they are a superposition of infinitely squeezed states. To physically construct GKP states, an envelope is applied on these infinite energy states to approximate them. A commonly used envelope is the Gaussian envelope $\exp(-\Delta^2\hat{n})$ (which can also be extended to multimode envelopes \cite{PRXQuantum.3.010335}) which on application to the ideal GKP codewords $\ket{\psi_{\mathrm{ideal}}}$ would yield a finite energy state. This state can be written as a superposition of finitely squeezed states \cite{PhysRevA.93.012315}. The envelope operation can be decomposed in the over-complete basis of displacement operations \cite{gottesman_encoding_2001, rozpkedek2021quantum} as a mixture of coherent displacement errors on the GKP states that can be treated as incoherent displacement noise of spread $\sigma=\tanh(\Delta^2/2)$ by a twirling argument involving the repeated application of stabilizers \cite{PhysRevA.101.012316,jafarzadeh2025logicalchannelsapproximategottesmankitaevpreskill,shaw2024logicalgatesreadoutsuperconducting}. Additionally, the bosonic pure-loss channel as well as thermal loss channel can be converted to a Gaussian displacement noise channel with the appropriate application of a quantum limited amplification channel \cite{8482307}. Hence, the displacement error noise model proves to be a very useful noise model to study.

\subsection{Achievable rates for communication of quantum information}\label{sec:achievable}
A rate $R$ is said to be achievable for a noise channel $\mathcal{N}$ if there exists a sequence of encoding and decoding operations $\mathcal{E}_n,\mathcal{D}_n$ such that over this sequence
\begin{equation}
\lim_{n\to\infty}\frac{\log_2(d_L(\mathcal{E}_n,\mathcal{D}_n))}{n}=R,\quad \lim_{n\to\infty}F_e(\mathcal{D}_n\circ\mathcal{N}^{\otimes n}\circ \mathcal{E}_n) = 1
\end{equation}
where $d_L$ is the total logical dimension of the encoding and decoding operations $\mathcal{E}_n,\mathcal{D}_n$ and $F_e$ is the channel fidelity defined as
\begin{equation}
F_e(\mathcal{Q}) = \langle\Phi|(\mathcal{Q}\otimes\mathcal{I}_A)(|\Phi\rangle\langle\Phi|)|\Phi\rangle   
\end{equation}
where $|\Phi\rangle$ is the purification of the maximally mixed state and $\mathcal{I}_A$ is identity on the ancillary system. This is a metric of how much entanglement survives through the channel and is also referred to as the entanglement fidelity \cite{tomamichel2016quantum} and has strong connections to average channel fidelity as well \cite{nielsen2002simple,PhysRevA.60.1888}. This essentially tells us that asymptotically, one can reliably send $R$ qubits of information per channel. The quantum capacity then follows naturally as the quantity $C_Q$ such that $R\leq C_Q$ if and only if $R$ is an achievable rate \cite{PhysRevA.55.1613}.

In this section, we will mainly be concerned with the coherent information of a channel which is an achievable rate that lower bounds the capacity of a channel \cite{PhysRevA.54.2629}. For a quantum channel $\mathcal{N}$, we can define the complementary channel $\mathcal{N}^c$ which describes the action of the channel on the environment if the channel is modeled as a unitary interaction with an environment. Defining the Stinespring dilation of $\mathcal{N}$ as
\begin{equation}
    \mathcal{N}(\rho_S) = \text{Tr}_E(U_{SE}(\rho_S\otimes|0\rangle\langle0|_E)U_{SE}^\dagger),
\end{equation}
where $U_{SE}$ is a unitary interaction of the system and environment, we get
\begin{equation}
    \mathcal{N}^c(\rho_S)=\text{Tr}_S(U_{SE}(\rho_S\otimes|0\rangle\langle0|_E)U_{SE}^\dagger).
\end{equation}

The coherent information of a channel $\mathcal{N}$ is then defined as
\begin{equation}
    I_c(\mathcal{N}) = \sup_{\rho}(I_c(\rho,\mathcal{N})),\quad I_c(\rho,\mathcal{N}) =S(\mathcal{N}(\rho)) - S(\mathcal{N}^c(\rho))
\end{equation}
where $I_c(\rho,\mathcal{N})$ can be defined for a specific choice of input state $\rho$. The coherent information quantifies the maximal amount of entropy that survives through a single use of the channel \cite{PhysRevA.54.2629}. For reliable information transmission through a channel, it is required that the environment doesn't measure out any of the information we wish to communicate. Hence through a single channel use, only as much entropy as $I_c$ can be sent through the channel. This then gives the definition for quantum channel capacity as
\begin{equation}
    C_Q(\mathcal{N}) = \lim_{n\to\infty}\frac{I_c(\mathcal{N}^{\otimes n})}{n}\geq I_c(\mathcal{N}),
\end{equation}
which due to the additivity of $S(\rho)$ gives the lower bound. 

We can link the entropic way of looking at information to an encoding operation by considering the following example. $A$ has $N$ bell pairs that they wish to share the halves of with $B$ through a noisy communication channel. If $A$ measures an appropriate number of their halves of the bell pair, they effectively create a stabilizer code for the halves of $B$ with the correct choice of unitary operations. Knowing that some amount of information is inevitably gained by the environment, the measurements by $A$ can be chosen in a way to ensure that the environment only gains redundant information that has already been measured by $A$. This can then allow the reliable communication (or equivalently the one-way distillation \cite{PhysRevA.54.3824}) of $K\leq N$ bell pairs effectively defining some error correcting code of form $[[N,K]]$.

\section{Analog information in concatenated GKP codes}
\subsection{Effective noise channel with analog information}\label{sec:logicalNoisesq}
In this section we provide details on the effective noise channel that acts on a GKP code which has been corrected using the closest lattice point decoding for the particular syndrome output. We will mainly concern ourselves with the case of a square GKP qudit in this section and refer the reader to Appendix \ref{app:CLPnoise} to see how this extends to general GKP codes.
We consider a single mode Gaussian displacement noise which can be described by the quantum channel
\begin{equation}
    \mathcal{N}_{\mathrm{displ}}(\cdot) =  \frac{1}{\pi \sigma^2}\int d^2\alpha e^{-|\alpha|^2/\sigma^2} D(\alpha)(\cdot)D(\alpha)^\dagger\label{eq:Ndispl}
\end{equation}
where $D(\alpha) = e^{\alpha a^\dagger - \alpha^{*}a}$. Since this is simply a 2-dimensional Gaussian random variable, we can equivalently treat this as independent Gaussian noise in the $\hat{q}$ and $\hat{p}$ quadratures respectively. Errors along $\hat{q}$ contribute to logical bit-flip and errors along $\hat{p}$ contribute to a logical phase-flip. Importantly, for a square GKP code, these are independent errors.\\
We now consider the event that we measure the syndrome $s_1$ for the value $\hat{q}\sqrt{\frac{d}{2\pi}}\mod1$ shifted to lie in $[-1/2,1/2)$. The smallest (in magnitude) displacement which gives syndrome $s_1$ is a shift of $s_1\sqrt{2\pi/d}$ along the $\hat{q}$ quadrature. Assuming we correct by exactly this much along the $\hat{q}$ direction, we are now left with an effective shift of $u\sqrt{2\pi/d}$ where $u$ is an integer which represents the logical error $\hat{X}^u$. We denote the probability of the event of logical $\hat{X}^u$ occurring after syndrome measurement $s_1$ as $p_1(u,s_1)$. To evaluate this probability, we need to find the set of displacements that cause a shift of $\hat{q}\to\hat{q}+e_1\sqrt{\frac{2\pi}{d}}$ such that $e_1\mod 1=s_1$ (hence giving a particular syndrome) and $(e_1-s_1)\mod d=u$ (hence giving the logical error $\hat{X}^u$). This restricts the values of $e_1$ to be
\begin{equation}
    e_1\in \{u+s_1+dl|l\in \mathbb{Z}\}
\end{equation}
while the overall set of displacements resulting in this syndrome would be given by $e_1\in\{s_1+l|l\in\mathbb{Z}\}$. Using the knowledge of the noise model, we know that $e_1\sqrt{2\pi/d}\sim\mathcal{N}(0,\sigma^2)$ which gives 
\begin{equation}
    p(u|s_1) = \frac{\sum_{l\in\mathbb{Z}}\exp(-\frac{\pi}{d\sigma^2}(dl+u+s_1)^2)}{\sum_{l\in\mathbb{Z}}\exp(-\frac{\pi}{d\sigma^2}(l+s_1)^2)}\label{eq:pu|s1}.
\end{equation}
Further, we can exactly find the probability distribution for the syndrome $p(s_1)$ to be
\begin{equation}
\begin{aligned}
    p(s_1)=\frac{1}{\sigma\sqrt{d}}\sum_l\exp(-\frac{\pi}{d\sigma^2}(l+s_1)^2),
\end{aligned}
\end{equation}
which lets us define
\begin{equation}
\begin{aligned}
    p_1(u,s_1)&=p(u|s_1)p(s_1)\\
    &= \frac{1}{\sigma\sqrt{d}}\sum_{l\in\mathbb{Z}}\exp(-\frac{\pi d}{\sigma^2}\left(l +\frac{u+s_1}{d}\right)^2),\label{eq:pus1}
\end{aligned}    
\end{equation}
showing that $p_1(u,s_1)$ is purely a function of $\frac{u+s_1}{d}$.  Since $u$ is necessarily modulo $d$ which means $u$ is an integer with $-\frac{d-1}{2}\leq u\leq \frac{d-1}{2}$ (assuming $d$ is odd). For now we will be restricting our discussion to $d$ being a prime number. Further, we know that $s_1$ is chosen to lie in the range $[-1/2,1/2)$ which means that the variable $\frac{u+s_1}{d}$ will satisfy $-\frac{1}{2}\leq\frac{u+s_1}{d}<\frac{1}{2}$.

The restriction of $\frac{u+s_1}{d}$ to this range by the appropriate definition of our modulo functions gives a clear understanding of what the summation over $l$ depicts in the expression of $p_1(u,s_1)$ in Eq. \eqref{eq:pus1}. The terms of $l\neq0$ are from the erroneous displacement being large enough that it displaces it outside the unit cell of the stabilizer lattice which is spaced by $\sqrt{2\pi d}$. Since the GKP code fundamentally has degenerate error sets, these can offer significant contributions depending on the form of the lattice \cite{Conrad2022gottesmankitaev}. However, in this case there is a very interesting behavior for a large value of $d$. The expression $p_1(u,s_1)$ shows a clear convergence towards the central term in the summation ($l=0$) which we will denote by
\begin{equation}
    p_\mathrm{lim}(u,s_1) = \frac{1}{\sigma\sqrt{d}}\exp(-\frac{\pi d}{\sigma^2}\left( \frac{u+s_1}{d}\right)^2),
\end{equation}
which is not a normalized distribution, but can be shown to be very close to the actual distribution for large enough $d$. To exactly quantify this, we show that on integrating $p_{\mathrm{lim}}(u,s_1)$ over the variables $u$ and $s_1$ (detailed calculation in Appendix \ref{app:dispSqGKP}), we get
\begin{equation}
    \sum_{|u|\leq\frac{d-1}2{}}\int ds_1 p_{\mathrm{lim}}(u,s_1) = 1-\mathcal{O}\left(\frac{\sigma}{\sqrt{d}}e^{-\frac{d\pi}{4\sigma^2}}\right),
\end{equation}
which gets very close to $1$ as $d/\sigma^2\gg1$. This means that the probability contribution from all the terms in the summation of $l\neq 0$ are exponentially suppressed. Additionally increasing $d$ makes $p_{\mathrm{lim}}$ have smaller spread around $\frac{u+s_1}{d}=0$, since it is proportional to a Gaussian in the variable $\frac{u+s_1}{d}$ with spread of $\sigma\sqrt{\frac{2\pi}{d}}$. This can be directly observed by evaluating the full function $p_1(u,s_1)$ numerically as is plotted in Fig. \ref{fig:1}(c).

We can similarly describe the probability distribution which relates the probability of a phase-flip error $\hat{Z}^v$ to the outcome of obtaining a syndrome measurement for $s_2$ denoted by $p_2(v,s_2)$. Similarly, if we consider an erroneous shift of $e_2$ such that $\hat{p}\to \hat{p}+e_2\sqrt{\frac{2\pi}{d}}$, we now have $e_2 = dl+v-s_2$ for $l\in\mathbb{Z}$. Since the random variable for $e_2$ follows the same distribution as $e_1$, it is easy to see that the distribution
\begin{equation}
\begin{aligned}
    p_2(v,s_2)&=p(u=v,s_1=-s_2)\\
    &= \frac{1}{\sigma\sqrt{d}}\sum_{l\in\mathbb{Z}}\exp(-\frac{\pi d}{\sigma^2}\left(l +\frac{v-s_2}{d}\right)^2),
\end{aligned}    
\end{equation}
which means that $p_2(v,s_2)$ is purely a function of $\frac{v-s_2}{d}$ with the same concentration behavior as $p_1(u,s_1)$. With the only difference being in the sign of the syndrome these distributions behave exactly the same. Putting the both of these together, we have a full characterization of the event of the error $\hat{X}^u\hat{Z}^v$ if we obtain the total syndrome $s_1,s_2$ which we denote by the probability distribution
\begin{equation}
    p(u,v,s_1,s_2)=p_1(u,s_1)p_2(v,s_2).
\end{equation}
Once the values of $s_1$ and $s_2$ are measured and we correct by the smallest displacement, we are left with Pauli noise on the qudit which has an error distribution of $p(u,v|s_1,s_2)$ for $\hat{X}^u\hat{Z}^v$. This can be equivalently be modeled by a qudit noise channel which also outputs $\ket{s_1,s_2}\in\mathcal{S}$ in an output register described by a Hilbert space $\mathcal{S}$ where $\langle s'_1,s'_2|s_1,s_2\rangle = \delta(s_1-s_1')\delta(s_2-s_2')$ which just ensures that we can exactly measure the values of $s_1$ and $s_2$. This is the same as the subsystem decomposition used for describing GKP states as a tensor product between logical Hilbert space and the stabilizer eigenvalues which can be used to define a complete basis for the bosonic mode \cite{PRXQuantum.5.010331,PhysRevA.107.062611,shaw2024logicalgatesreadoutsuperconducting}. We can use this to define an effective logical noise channel $\mathcal{N}_{\mathrm{logical}}:\mathcal{L}(\mathcal{H}_d)\to\mathcal{L}(\mathcal{H}_d\otimes \mathcal{S})$ defined as
\begin{equation}
\begin{aligned}
    \mathcal{N}_{\mathrm{logical}}(\cdot) = \sum_{u,v}\int d&s_1 ds_2 p(u,v,s_1,s_2)\\&\hat{X}^u\hat{Z}^v(\cdot)(\hat{X}^u\hat{Z}^v)^\dagger
    \otimes|s_1,s_2\rangle\langle s_1,s_2|
\end{aligned}\label{eq:logicalN}
\end{equation}
which once we measure out the extra register reduces to qudit Pauli noise with distribution $p(u,v|s_1,s_2)$. Since the extra register is effectively classical in the sense that it can only give classical information which is the syndrome itself, this noise channel is a classical mixture of Pauli noise with an auxiliary system. We will now examine an achievable rate for this channel which fundamentally makes use of the analog information. 

\begin{figure*}[ht]
    \centering
    \includegraphics[width=\linewidth]{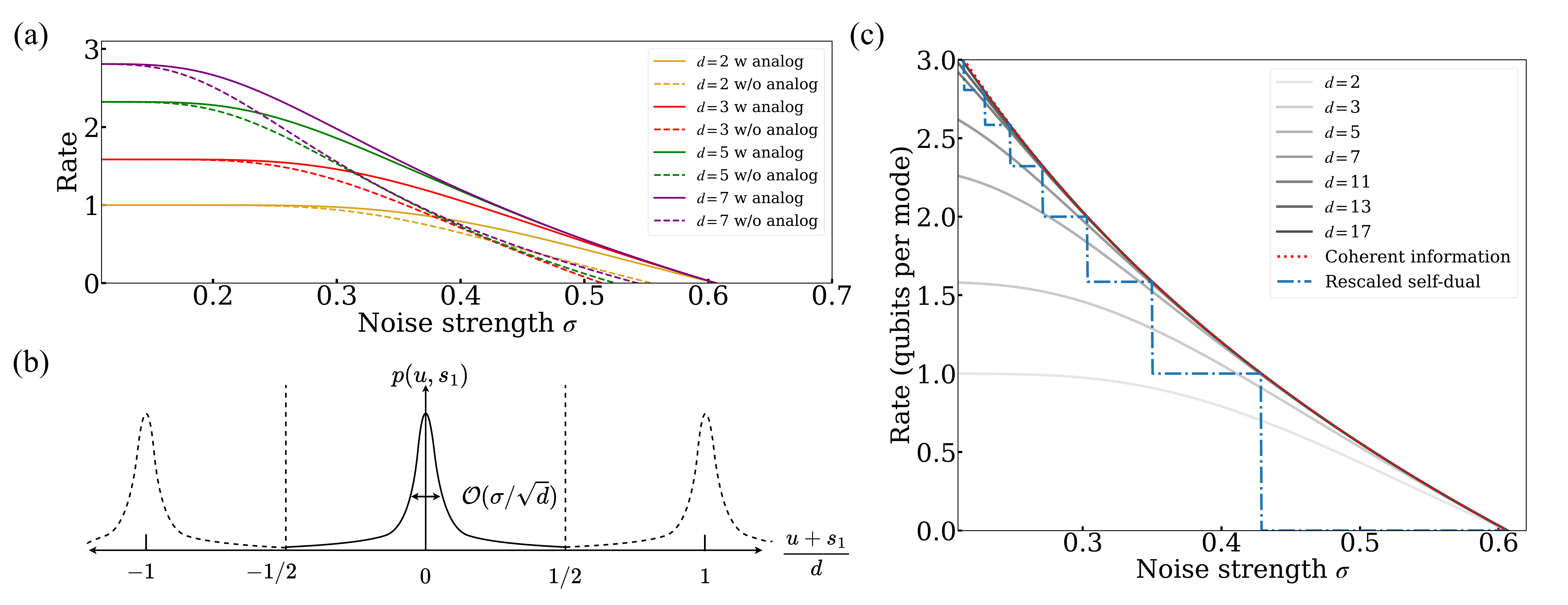}
    \caption{(a) Achievable rates for the square GKP qudit with analog information ($I^{\mathrm{sq}}_{d,\mathrm{analog}}$) and without analog information ($I^{\mathrm{sq}}_{d,\text{no analog}}$) plotted in solid line and dashed lines respectively for different values of $d$. The dashed line uses the distribution $p(u)$ obtained by averaging over the syndrome hence resulting in not using any analog information. (b) A rough depiction of what the concentration of $p_1(u,s_1)$ looks like in the limit of large $d$. While there is contribution from terms that are centered around $\pm l$ for $l=1,2\dots$ (see plotted Gaussians with dashed lines), their contribution is exponentially small in magnitude. Additionally, the spread of $p_{\mathrm{lim}}(u,s_1)$ being $\mathcal{O}(\sigma/\sqrt{d})$ over the normalized variable $(u-s_1)/d$ results in the size of typical error set growing much smaller than the actual dimension of Hilbert space. (c) Numerical evaluation of $I^{\mathrm{sq}}_{d,\mathrm{analog}}$ (for values of $d=2$ to $d=17$ shown in solid lines) showing a clear convergence to the coherent information of the Gaussian displacement noise channel (dotted line). Notably for finite and reasonably small $d$ ($<10$), we already see the value of $I^{\mathrm{sq}}_{d,\mathrm{analog}}$ get numerically close to coherent information while also crossing the value of rates achieved by rescaled self-dual lattices (dashed-dotted line).}
    \label{fig:2}
\end{figure*}
\subsection{Achievable rate using analog information}\label{sec:Isqanalog}

For a qudit Pauli noise channel with noise distribution $p_{u,v}$ for the error $\hat{X}^u\hat{Z}^v$, we can find the coherent information for the input of a maximally mixed state to be
\begin{equation}
    I_c(\mathbb{I}/d,\mathcal{N}_{\mathrm{pauli}}) = \log_2(d) - H_2(p_{u,v})
\end{equation}
where $H_d(p_{u,v}) = -\sum_{u,v}p_{u,v}\log_d(p_{u,v})$. This gives the Hashing bound for Pauli channels which is achievable through the use of random CSS codes \cite{PhysRevA.54.1098} hence is necessarily below the capacity of the Pauli channel \cite{PhysRevA.78.062335}. Similarly for a classical mixture of Pauli (CMP) channels with an auxiliary system defined as
\begin{equation}
    \mathcal{N}_{\mathrm{CMP}}(\cdot) = \sum_{u,v,j}p_{u,v,j}\hat{X}^u\hat{Z}^v(\cdot)(\hat{X}^u\hat{Z}^v)^\dagger\otimes|j\rangle\langle j|,
\end{equation}
taking the input to be the maximally mixed state yields the coherent information to be
\begin{equation}
    \begin{aligned}
        I_c(\mathbb{I}/d,\mathcal{N}_{\mathrm{CMP}}) &= \log_2(d)-\sum_{j}p_j\log_2(p_j)\\&\quad\quad\quad\quad + \sum_{u,v,j}p_{u,v,j}\log_2(p_{u,v,j})\\
        &= \log_2(d) - H_2(u,v|j)
    \end{aligned}
\end{equation}
where $H_d(u,v|j) = H_d(p_{u,v,j}) - H_d(p_j)$ is the conditional information entropy of $u,v$ conditioned on $j$ and $p_j=\sum_{u,v}p_{u,v,j}$. If the register containing information on $j$ was traced out, we obtain a Pauli channel with distribution $p_{u,v}$ for which the coherent information will be less than $I_c(\mathbb{I}/d,\mathcal{N}_{\mathrm{CMP}})$ by the mutual information $I(u,v;j)$ which is always non-negative.

From this we can see that previous methods for calculating achievable rates associated to the square GKP qudit \cite{PhysRevA.64.062301,Raveendran2022finiterateqldpcgkp} were in essence calculating the rate associated to the CMP obtained by not having access to the auxiliary system for the channel in Eq. \eqref{eq:logicalN}. Hence we now define an achievable rate associated to square GKP qudits as $I_c(\mathbb{I}/d,\mathcal{N}_{\mathrm{logical}})$ which equals
\begin{equation}
    I_{d,\mathrm{analog}}^{\mathrm{sq}} = \log_2(d) + 2\int ds_1 \sum_{u=0}^{d-1} p_1(u,s_1)\log_2\left(\frac{p_1(u,s_1)}{p(s_1)}\right),
\end{equation}
where we use the fact that $p(u,v,s_1,s_2)$ fully separates into $p_1(u,s_1)p_2(v,s_2)$ and $H(u|s_1)=H(v|s_2)$. These are evaluated for different values of underlying noise strength $\sigma$ and plotted in Fig. \ref{fig:2}(a) in comparison to the coherent information of the Pauli noise channel with probability distribution $p(u,v)$ obtained by throwing out syndrome information. In the numerical evaluation of $I_{d,\mathrm{analog}}^{\mathrm{sq}}$ we observe that the threshold value of $\sigma$ beyond which the channel does not admit a rate seems independent of the value of $d$. This threshold value can be numerically checked to be approximately $\frac{1}{\sqrt{e}}$ which has been observed in concatenated GKP schemes that use analog information \cite{lin2024exploringquantumcapacitygaussian,PhysRevLett.119.180507}.

As $d$ is increased, there seems to be a clear convergence of the value that $I^{\mathrm{sq}}_{d,\mathrm{analog}}$ takes independent of $d$. In particular we find that for any given $\sigma$, once $d\gg\sigma^{-2}$, the value
\begin{equation}
    \left|I^{\mathrm{sq}}_{d,\mathrm{analog}} - \log_2\left(\frac{1}{\sigma^2 e}\right)\right|\leq\mathcal{O}(e^{-d\pi\sigma^2}) + \mathcal{O}(d^{1/2}\sigma^{-1}e^{-\frac{\pi d}{9\sigma^2}}),
\end{equation}
the derivation of which is detailed in Appendix \ref{app:dispSqGKP}. This is an incredibly interesting consequence since it is known from \cite{holevo1999evaluatingcapacitiesbosonicgaussian,PhysRevA.64.062301} that the coherent information of the Gaussian displacement noise channel $\mathcal{N}_{\mathrm{displ}}$ (see Eq. \eqref{eq:Ndispl}) is
\begin{equation}
    I_c(\mathcal{N}_{\mathrm{displ}}) =\log_2\left(\frac{1}{\sigma^2 e}\right).
\end{equation}
Hence if we choose a large enough $d$ and concatenate it to some outer code capable of achieving $I_{d,\mathrm{analog}}^{\mathrm{sq}}$, we can get very close to $I_c(\mathcal{N}_{\mathrm{displ}})$ as an achievable rate. In fact by taking $d\to\infty$ we can concretely claim $I_c(\mathcal{N}_{\mathrm{displ}})$ as achievable using this particular choice of codes.

The reason of this convergence is roughly the same as the reason $p_1(u,s_1)$ converges to $p_{\mathrm{lim}}(u,s_1)$. The distribution of $p_{\mathrm{lim}}(u,s_1)$ implies that on increasing $d$, only the contributions of $X^u$ errors with $|u|< \mathcal{O}(\sqrt{d})$ become relevant since the other probabilities are exponentially suppressed. This means that even though the Hilbert space dimension has increased by increasing $d$, the size of the typical error set is only growing as a function of $\sqrt{d}$. The displacements along $\hat{q}$ that result in the same syndrome are spaced apart by $\sqrt{\frac{2\pi}{d}}$ and as a result increasing $d$ ends up making all the syndromes equally likely. As a result, the marginal distribution $p(s_1)$ approaches the uniform distribution over $s_1$ between $-1/2$ to $1/2$ differing only in $\mathcal{O}(e^{-\pi d\sigma^2})$. However, it is clear that even in this limit, the syndrome is providing useful information since $p_1(u,s_1)$ is a function of $\frac{u+s_1}{d}$. The convergence toward $p_{\mathrm{lim}}$ lets us analytically approximate the value of $I^{\mathrm{sq}}_{d,\mathrm{analog}}$ while bounding the difference from the actual values in functions that are exponentially suppressed in increasing the value of $d/\sigma^2$. Putting all of this together, we reach our first claim that through the concatenation of appropriate outer code with square GKP qudits, the coherent information is indeed an achievable rate for all values of the noise strength $\sigma$.

Numerically, we note that since the difference in rates is exponentially suppressed in growing $d$, even for $d=5$ and $\sigma>0.4$, the achievable rate is numerically very close to $I_c(\mathcal{N}_{\mathrm{displ}})$. In the following section we will show how quantum polar codes can be shown to achieve $I^{\mathrm{sq}}_{d,\mathrm{analog}}$ which shows that we can explicitly construct the GKP codes that have rates very close to $I_c(\mathcal{N}_{\mathrm{displ}})$.

\section{Polar codes as a candidate for concatenation}
\subsection{Channel polarization for classical noise}\label{sec:classical_polar}
We introduce a notation of $u^N_1$ which represents a vector of length $N$ with elements labeled as $u_i$, and a subset of it $u^j_i$ is the terms of indices between $i$ to $j$ (inclusive of both). We now motivate the definition of a discrete memoryless channel (DMC). A classical noise channel $W$ can be seen as a map between symbols in an input alphabet $\mathcal{X}$ to an output alphabet $\mathcal{Y}$ where the probability symbol $x$ maps to $y$ is given by $W(y|x)$. The memoryless aspect of this noise channel means that if this were to act individually over $N$ different symbols $x^N_1$ where each $x_i\in\mathcal{X}$, the output symbol $y_i$ only depends on $x_i$ and has no memory of the symbols that came before it. Hence the noise acts independently on each symbol. Common discrete memoryless channels include erasure channels, bit-flip channels and also the additive white Gaussian noise (AWGN) channel. Here we are concerned with dits so $\mathcal{X}\equiv \mathbb{F}_d$ and we will not place any limitations on what the output alphabet $\mathcal{Y}$ is. One can always find the symmetric capacity of this channel which is given by
\begin{equation*}
    I(W) = \sum_{x\in\mathcal{X}}\sum_{y\in\mathcal{Y}}\frac{1}{d}W(y|x)\log_d\left(\frac{W(y|x)}{\sum_{x'\in\mathcal{X}}\frac{1}{d}W(y|x')}\right),
\end{equation*}
and will always lie between 0 and 1. This represents the rate (in dit per channel use) of information transfer that can be achieved through this channel assuming all inputs to be equally likely.  

Let us suppose we have a linear transform $G_N:\mathbb{F}_d^{N}\to\mathbb{F}_d^N$ and it acts on a dit-string of length $N$ which we represent as $u^N_1$, and maps it to $x^N_1$. We now communicate $x^N_1$ over $N$ copies of the noise channel $W$ which maps it to the output $y^N_1$. Consider the task of figuring out $u^N_1$ using the output $y^N_1$ aided with a genie which supplies the values $u^{i-1}_1$ when we wish to estimate $u_i$. This by itself can be understood by treating the whole set $y^N_1,u^{i-1}_1$ as an output for the dit $u_i$ whilst all of $u_{i+1}^N$ can take any possible value. Note that this is nothing but a discrete memoryless channel of its own mapping $u_i$ to the tuple $(y^N_1,u^{i-1}_1)$ with some probability $W^{(i)}((y^N_1,u^{i-1}_1)|u_i)$. Wherever we use the notation $W^{(i)}$, we will be referring to the channel $W^{(i)}((y^N_1,u^{i-1}_1)|u_i)$. The channel $W^{(i)}$ has a very interesting property of its own. As the block-length $N$ increases, the fraction of channels that have $I(W^{(i)})\in(1-\delta,1]$ begins to tend to $I(W)$ whilst the fraction of channels with $I(W^{i})\in[0,\delta)$ tends to $1-I(W)$ for any $\delta>0$ for a particular choice of $G_N$, namely the polar transform (described in figure \ref{fig:polartransform}) \cite{sasoglu2009polarizationarbitrarydiscretememoryless,Arikan_2010}. In essence, all these channels either are nearly perfect for information transfer or are completely useless which is where the term polarization arises from.

Another relevant channel parameter is the Bhattacharya distance $Z(W)$ defined as
\begin{equation}
    Z(W) = \frac{1}{d(d-1)}\sum_{x,x'\in\mathcal{X},x\neq x'}\sum_{y\in\mathcal{Y}}\sqrt{W(y|x)W(y|x')},
\end{equation}
which also happens to give an inequality for the capacity as
\begin{equation}
\begin{aligned}
    \log_d\left(\frac{d}{1 + (d-1)Z(W)}\right) &\leq I(W)\\
    &\leq 1- \frac{(1-\sqrt{1-Z(W)^2})}{\log_2(d)},
\end{aligned}
\end{equation} 
proven in \cite{sasoglu2009polarizationarbitrarydiscretememoryless}. The importance of this is that $I(W)\to1$ if and only if $Z(W)\to 0$. This has a direct consequence in the way we design a polar code since the values of $Z(W)$ can be estimated through Monte-Carlo sampling since it is essentially the sum of the averages of $\sqrt{\frac{W(y|u)}{W(y|u')}}$ for all pairs $u\neq u'$ \cite{5075875}.

A polar code is defined in terms of the set of symbols that are chosen to be frozen before the encoding procedure. The set $\mathcal{A}\subseteq\{1,\dots N\}$  is the set of symbols that are used for the encoding procedure. The error probability of a polar code under a successive cancellation (SC) decoder can be upper bounded by
\begin{equation}
    P_e\leq (d-1)\sum_{i\in\mathcal{A}}Z(W^{(i)}),
\end{equation}
which follows from the fact that the SC decoder is sequentially doing the maximum likelihood estimate of $u_i$ using $W^{(i)}$, a task for which the error probability is bounded above by $(d-1)Z(W^{(i)})$ \cite{sasoglu2009polarizationarbitrarydiscretememoryless}. This is notably not the same as the overall maximum likelihood decoding since the channels $W^{(i)}$ only use knowledge of the symbols $u^{i-1}_1$ and assume even the possibly frozen symbols in $u^N_{i+1}$ to be random. The asymptotic behavior of $Z(W^{(i)})$ is proven in \cite{6357295} where it is shown that the fraction of $W^{(i)}$ with $Z(W^{(i)})\leq 2^{-N^\beta}$ for any $\beta<1/2$ approaches the value of $I(W)$ as $N\to\infty$. This hence gives a sequence of codes with $P_e\to0$ and $\frac{|\mathcal{A}|}{N}\to I(W)$ hence being capacity achieving. For a fixed $N$, we refer to the `good' indices $i$ as those with $Z(W^{(i)})<\delta$ and the `bad' indices to be everything else. The good indices will go in $\mathcal{A}$, hence having $u_i$ carry information and all the other $u_i$ with $i\in\mathcal{A}
_c$ will be frozen to a particular symbol. We can choose a sufficiently small $\delta$ to ensure that the channels $W^{(i)}$ are sufficiently good enough for information transmission.

We refer the reader to Appendix \ref{app:polarencoder} for a description of the encoding procedure and Appendix \ref{app:polarSCdecoder} for the SC decoding procedure for a polar code. We make use of the fact that polarization is known to happen for channels with input alphabets that are prime in number \cite{sasoglu2009polarizationarbitrarydiscretememoryless} as well as the fact that the kernel used for the encoding operation can be chosen to be any matrix of the form
\begin{equation}
    G_N = \begin{pmatrix}1 & \alpha\\
    0 & 1
    \end{pmatrix}^{\otimes n},
\end{equation}
where $\alpha\in\mathbb{F}_d\backslash\{0\}$ \cite{5513568}. We now proceed to show how the principle of channel polarization can be used to achieve the rate $I_{d,\mathrm{analog}}^{\mathrm{sq}}$.
\begin{figure}[h]
    \centering
    \includegraphics[width=\linewidth]{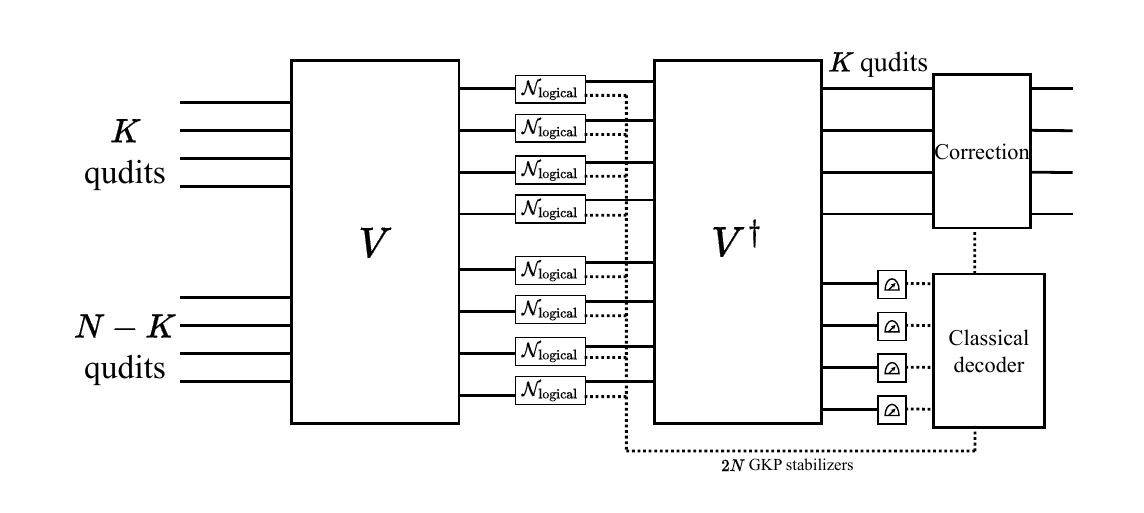}
    \caption{A heuristic circuit diagram explaining the structure of the concatenation. Here we treat each single mode square GKP as its own qudit which undergoes an encoding operation $V$ which creates the codewords for the concatenated outer polar code. The operations in $V$ are logical qudit operations (in this case purely composed of CSUM gates in the case of the polar code). The $K$ qudits are used for encoding information and the $N-K$ qudits are frozen to some chosen state. The logical noise operation $\mathcal{N}_{\mathrm{logical}}$ acts on all the encoded qudits and also outputs classical information in the form of the analog GKP syndrome of each individual mode. Following this the decoding operation proceeds, which involves inverting the encoding operation and then measuring the previously frozen $N-K$ qudits which gives the syndrome information of the outer polar code. This along with all the individual GKP syndromes is then used in a classical decoder (the successive cancellation decoder) giving a correction to be applied on the $K$ qudits carrying information.}
    \label{fig:3}
\end{figure}

\subsection{CSS-like entanglement assisted polar code construction}\label{sec:css_polar}
Given their capacity achieving properties in the case of classical memoryless noise channels, it was only natural that there was work done to extend this to quantum noise channels. Wilde and Guha showed that quantum polar codes can be constructed for achieving the capacity of classical-quantum channels \cite{6302198} as well as the coherent information of degradable channels with classical environment \cite{6302198}, albeit without an efficient decoding method. Dupuis et al \cite{PhysRevLett.109.050504} used observed that a single qubit Pauli noise channel can be treated as a combination of two DMCs with one representing bit-flip noise and the other representing phase-flip (conditioned on whether a bit-flip has occurred or not), a CSS-like entanglement assisted code can be constructed using two classical polar codes. This code achieves a net coding rate equal to the Hashing bound of this noise channel. Importantly this scheme has both efficient encoding and decoding since the encoding circuit consist of just CNOT operations and the decoder simply reverses the encoder and then the syndrome can be fed into a classical decoder. This was further generalized to beyond Pauli channels in \cite{8989387,9366784} showing a channel splitting and recombining method using random two-qubit Clifford operations. These have also been extended for qudit input channels as shown in \cite{9517845}. We will mainly be using the principles as they are presented in \cite{PhysRevLett.109.050504} and extending them for qudit CMP channels. We refer the reader to \cite{goswami2021quantum} for a more detailed treatise on quantum polar codes.

We first begin by examining what the polar transform looks like over qudits. We wish to be able to map all the $N$ qudit states $\ket{\mathbf{z}}$ ($\mathbf{z}\in\mathbb{F}^N_d$) chosen to be the computational basis to $\ket{G_N\mathbf{z}}$ and this is the basis where $\hat{Z}$ is diagonal. The diagonal basis for $\hat{X}$ will be represented by $\ket{\tilde{\mathbf{x}}}$ ($\mathbf{x}\in\mathbb{F}_d^N$). Since a single unit of the encoding operation takes $(z_1,z_2)\to(z_1+\alpha z_2,z_2)$ for $z_1,z_2\in\mathbb{F}_d$, this operation is equivalent to $\alpha$ repetitions of the CSUM gate. Extending this to the quantum case we can define the unitary
\begin{equation}
    \hat{V} = \sum_{\mathbf{z}\in F_d^N}|G_N\mathbf{z}\rangle\langle\mathbf{z}| = \sum_{\mathbf{x}\in F_d^N}|G_N^{-T}\tilde{\mathbf{x}}\rangle\langle\tilde{\mathbf{x}}|,\label{eq:polarV}
\end{equation}
where we can note that since $V$ is obtained by the use of $\mathcal{O}(N\log(N))$ CSUM gates, each gate in the $X$ basis becomes a CSUM$^{-1}$ gate with the control and target swapped. Hence effectively the polar transform in the $\hat{X}$ basis is given by
\begin{equation}
    G_N^{-T} = \begin{pmatrix}
        1 & 0\\
        -\alpha & 1
    \end{pmatrix}^{\otimes n},
\end{equation}
which albeit is a different kernel, it will still give channel polarization as shown in \cite{5513568}. It should also be noted that the order of the dits is effectively reversed in the $x$ basis.
\begin{figure*}[ht]
    \centering
    \includegraphics[width=0.85\linewidth]{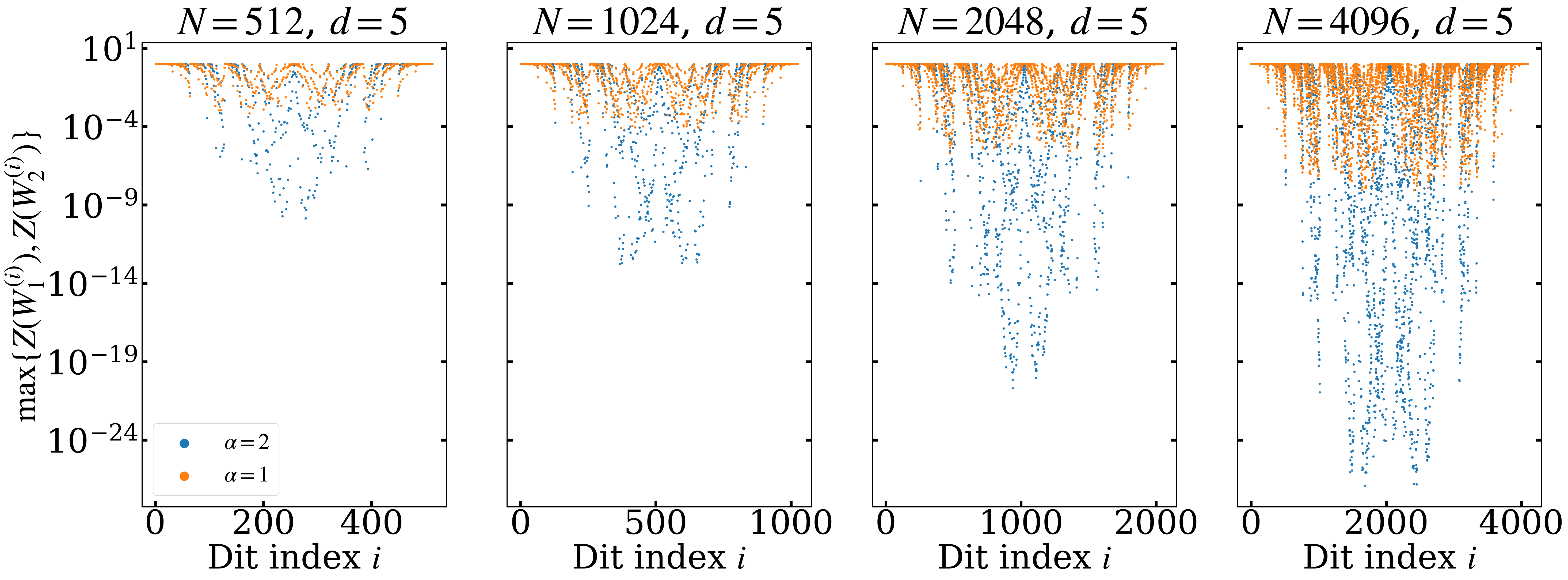}
    \caption{Effects of channel polarization on increasing block length $N$. Here we simulate the values of $Z(W_1^{(i)})$ and $Z(W_2^{(i)})$ for displacement spread $\sigma=0.4$ using kernel with $\alpha=2$ (blue dots) and $\alpha=1$ (orange dots) for $d=5$. Information will be encoded in dits that have both $Z(W_1^{(i)})$ and $Z(W_2^{(i)})$ being small hence we plot $\max\{Z(W_1^{(i)}),Z(W_2^{(i)})\}$. As can be seen for small block lengths, not enough of the indices have both the values being small as can be seen for $N=512$, and as can be seen in increasing value of $N$ till finally reaching $N=4096$. Notably, the polarization phenomenon is observed to be quicker using the kernel of $\alpha=2$ compared to $\alpha=1$ as can be inferred by comparing the fraction of indices satisfying $\max\{Z(W_1^{(i)}),Z(W_2^{(i)})\}<\delta$ for some small enough $\delta$.}
    \label{fig:4}
\end{figure*}
As was observed in \cite{9517845,PhysRevLett.109.050504}, the qudit Pauli noise channel with $p_{u,v}$ probability for the operation $\hat{X}^u\hat{Z}^v$ can be understood as two DMCs
\begin{equation}
\begin{aligned}
    W_A(z+u|z)&=\sum_vp_{u,v},\\
    W_P((x+v,u)|x)&=p_{u,v},
\end{aligned}
\end{equation}
where $W_A$ represents bit-flip noise and $W_P$ represents phase-flip noise. Notably these are not independent noise channels in general and so to account for that, we represent the phase noise as a DMC which also outputs the dit $u$ indicating the amount of dit-flip that has occurred already since this conditions 
the probability of a phase-flip after that. We now consider the classical polar codes that can achieve the capacity of $I(W_A)$ and $I(W_P)$ in dit per channel use (or equivalently $\log_2(d)I(W_A)$ and $\log_2(d)I(W_P)$ in bit per channel use) and note that we can simultaneously achieve their capacity since the polar transform $\hat{V}$ in Eq. \eqref{eq:polarV} does the polar transform $G_N$ in the $Z$ basis and the polar transform $G_N^{-T}$ in the $X$ basis. To be able to communicate quantum information through a particular index $i\in\{1,\dots N\}$, we require it to be both a `good' index against bit-flip and phase-flip noise. Since the polarization in the phase basis occurs in the reversed order of indices, this means we would always have to consider 4 disjoint sets which an index $i$ must lie in which are as follows.
\begin{itemize}
    \item $\mathcal{I}$: set of indices good against both bit and phase-flip noise.
    \item $\mathcal{A}$: set of indices good against phase-flip noise but bad against bit-flip noise.
    \item $\mathcal{P}$: set of indices good against bit-flip noise but bad against phase-flip noise.
    \item $\mathcal{E}$: set of indices bad against both noise.
\end{itemize}
Naturally this means that we can use qudits with index $i\in\mathcal{I}$ to send quantum information. Qudits with index $i\in\mathcal{A}$ will be required to be frozen in the $Z$ basis and qudits with index $i\in\mathcal{P}$ will be required to be frozen in the $X$ basis. Notably due to the structure of $\hat{V}$ being only CSUM gates, $\hat{V}\hat{X}_i\hat{V}$ and $\hat{V}\hat{Z}_i\hat{V}$ always stay as $\hat{X}$ or $\hat{Z}$ strings. Since $\mathcal{A}$ and $\mathcal{P}$ are non-intersecting, freezing their indices will also make sure we end up with commuting stabilizers which are either purely made of Pauli $X$ or Pauli $Z$ operations giving a CSS like construction. The complication arises when we consider $\mathcal{E}$ which must be frozen in both $X$ and $Z$ bases simultaneously. Using an entanglement assisted coding method \cite{brun2016entanglementassistedquantumerrorcorrectingcodes,PhysRevA.66.052313,bennett2002entanglementassistedcapacityquantumchannel}, these can then be chosen to be shared bell pairs between the sender and receiver which would make this a valid coding method. The set $\mathcal{A}\cup\mathcal{E}$ would be frozen in the $Z$ basis (hence $|\mathcal{A}\cup\mathcal{E}| \approx N(1-I(W_A))$) and the set $\mathcal{P}\cup\mathcal{E}$ would be frozen in the $X$ basis (hence $|\mathcal{P}\cup\mathcal{E}| \approx N(1-I(W_P))$), and here $\mathcal{E}$ is frozen in both bases by the fact that it is shared bell pairs. As a result of the possible entanglement assistance, the net coding rate (in qudit per channel use) this achieves is given by
\begin{equation}
    \frac{|\mathcal{I}|-|\mathcal{E}|}{N}\to I(W_A)+I(W_P)-1 = 1-H_d(p_{u,v}),
\end{equation}
where we detail the encoding and decoding of this in Appendix \ref{app:CMP}.

Since we are working with a classical mixture of Pauli noise Eq. \eqref{eq:logicalN}, we can consider the syndrome to be part of the noisy output which yields two noise channels
\begin{equation}
\begin{aligned}
    W_1(z+u,s_1|z) &= p_1(u,s_1),\\
    W_2(x+v,s_2|x) &= p_2(v,s_2),
\end{aligned}\label{eq:W1W2}
\end{equation}
that are still DMC, and will obey all the polarization requirements that $W_A$ and $W_P$ do. Hence we similarly see that
\begin{equation}
    \frac{|\mathcal{I}|-|\mathcal{E}|}{N}\to I(W_1)+I(W_2)-1 = \frac{I_{d,\mathrm{analog}}^{\mathrm{sq}}}{\log_2(d)},
\end{equation}
which essentially means that if we concatenate a quantum polar code constructed using the knowledge of noise channels $W_1$ and $W_2$, we can obtain an explicit concatenated square GKP code capable of achieving $I_{d,\mathrm{analog}}^{\mathrm{sq}}$. The success of the code relies on a decoding task which is the combination of decoding the bit and phase errors separately. Hence we must consider the failure probabilities of both decoders
\begin{equation}
\begin{aligned}
    P_{e,1}\leq(d-1)\sum_{i\in\mathcal{I}\cup\mathcal{P}}Z(W_1^{(i)}),\\
    P_{e,2}\leq(d-1)\sum_{i\in\mathcal{I}\cup\mathcal{A}}Z(W_2^{(i)}),
\end{aligned}\label{eq:polarineq}
\end{equation}
which must both go to zero, which is ensured by the individual polarization of $W_1$ and $W_2$. It is important to point out that we only make this claim for prime values of $d$ since this result hinges on the ability that arbitrary kernels can be used for polarization as well as the polarization sharply splits channels into having $I(W)\to1$ or $I(W)\to0$, which is not necessarily true for all kernels if non-prime values of $d$ were to be used \cite{sasoglu2009polarizationarbitrarydiscretememoryless}. While we expect these results to be generalizable to arbitrary $d$, we leave that exploration for future work.
\subsection{Numerical results and discussion}\label{sec:numeric_polar}
Designing a polar code requires being able to find the values of $Z(W^{(i)})$ to be able to judge which indices are good and bad for information communication. This is general not an easy task since this would involve summing over an exponentially increasing number of terms as can be observed from the fact that the output alphabet of $W^{(i)}$ lies in $\mathcal{Y}^N\times\mathcal{X}^{i-1}$. This hence calls for a sampling based method to be able to estimate $Z(W^{(i)})$. We note that the channels we are working with are symmetric, which is to say that for any $a\in\mathbb{F}_d$, there is a permutation $\pi_a:\mathcal{Y}\to\mathcal{Y}$ such that $W(y|x)=W(\pi_a(y)|x\oplus a)$. This is trivially true for the channels $W_1$ and $W_2$ in \ref{eq:W1W2}. Due to this symmetry, we only need to be concerned with the effect of these channels on the all $0$ input and being able to decode it back to all $0$ perfectly. Also, this makes the choice of frozen values to not have any consequence on code performance, hence we can choose the frozen values $u_{\mathcal{A}_c}$ to always be $0$. Due to this symmetry, we note that $Z(W)$ depends only on the values of $\langle\sqrt{\frac{W(y|a)}{W(y|0)}}\rangle_y$ averaged over all possible outputs of $W$. The SC decoder in fact finds these values for the channels $W^{(i)}$ during the decoding procedure which it then uses to make the decision of what the best estimate $\hat{u}_i$ is for the symbol $u_i$. The closer this expression is to $0$, this reflects more sureness for $u_i =0$ which is the correct value, and the closer this is to $1$ reflects less sureness. Hence we pick $M$ samples from the distribution the outputs follow and then in $\mathcal{O}(MN\log(N))$ time we are able to estimate all the values of $Z(W^{(i)})$ with relative accuracy $\mathcal{O}(M^{-{1/2}})$ using Monte-Carlo estimation.
\begin{figure}[h]
    \centering
    \includegraphics[width=\linewidth]{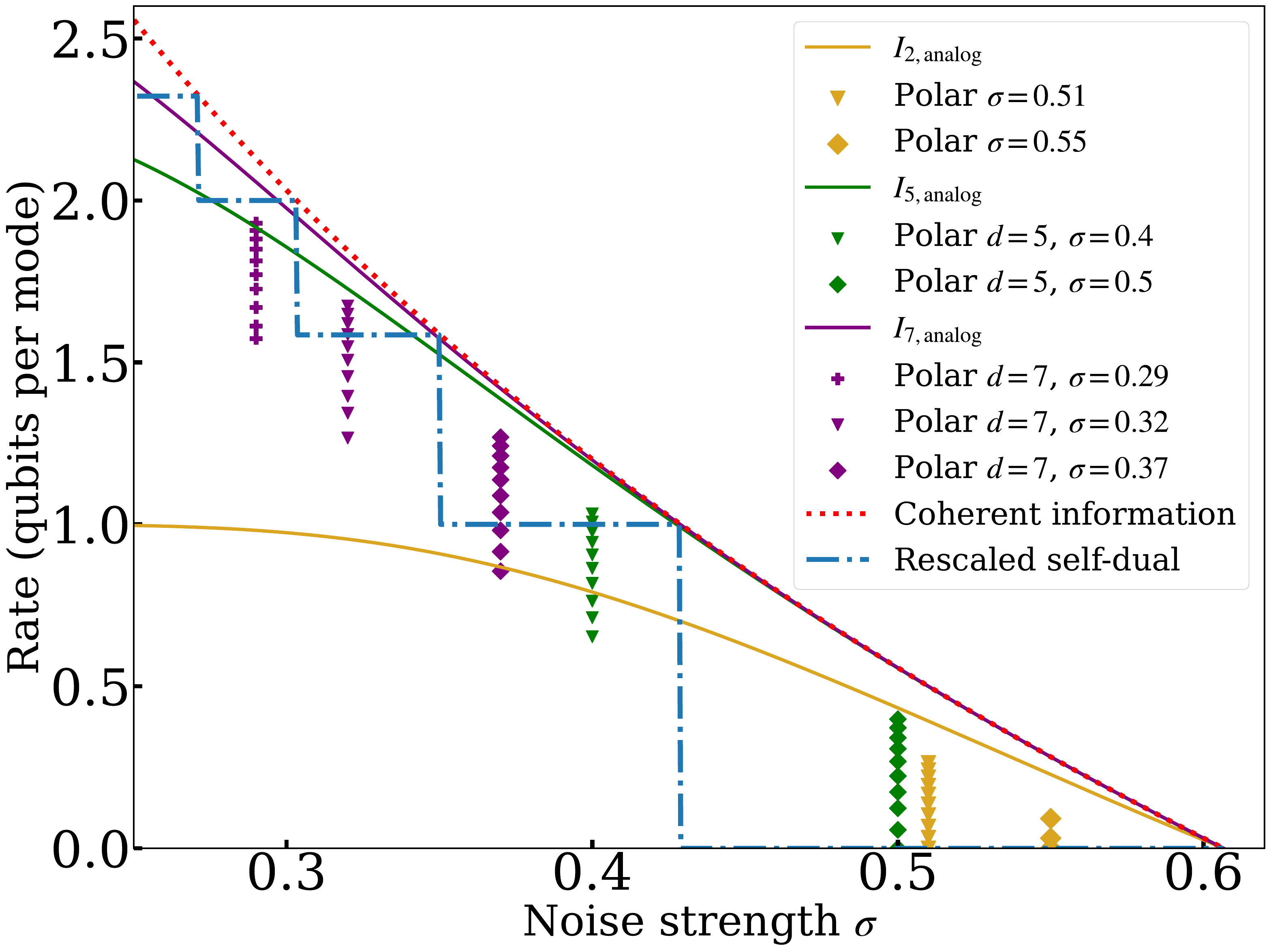}
    \caption{Performance of polar codes concatenated to square GKP qudits of $d=2$, $d=5$, and $d=7$. Here we find an explicit sequence of codes with increasing $N$ which satisfies $P_{e,1}\leq c_e N^{-\beta}$ and $P_{e,2}\leq c_e N^{-\beta}$ at a given value of $\sigma$ as the value of $N$ is increased. For the above plot we use $c_e=0.5$ and $\beta=-2/9$. The kernel choice for $d=2$ is $\alpha=1$, $d=5$ is $\alpha=2$, and for $d=7$ is $\alpha=3$. We find that we are clearly able to approach the solid line representing $I_{d,\mathrm{analog}}^{\mathrm{sq}}$ for different values of $d$ which in turn is able to exceed the rate achievable using rescaled self-dual lattices (dashed-dotted line).}
    \label{fig:5}
\end{figure}
Note that we need both $Z(W_1^{(i)})$ and $Z(W_2^{(i)})$ to simultaneously approach $0$ for $i$ to be a good index. We can see channel polarization in action in Fig. \ref{fig:4} which shows the estimated values of $Z(W_1^{(i)})$ and $Z(W_2^{(i)})$ for $N=512$ to $4096$ for a square qudit ($d=5$) with $\sigma=0.4$ and kernel choice $\alpha=2$ and $\alpha=1$ in blue and orange respectively. The fraction of indices satisfying both $Z(W_1^{(i)})<\delta$ and $Z(W_2^{(i)})<\delta$ can be clearly seen to be increasing with $N$ as one would expect for both choices of $\alpha$. 

We further construct sequences of polar codes satisfying
\begin{equation}
    P_{e,1}\leq c_e N^{-\beta},\quad P_{e,2}\leq c_e N^{-\beta},
\end{equation}
for a particular value of noise strength $\sigma$ and constants $0<c_e<1$ and $\beta>0$. We do this by finding the best choices of $\mathcal{I},\mathcal{A},\mathcal{P}$ and $\mathcal{E}$ based on the estimated values of $Z(W_1^{(i)})$ and $Z(W_2^{(i)})$ and use Eq. \eqref{eq:polarineq}. We report our results for $d=2$, $d=5$ and $d=7$ in Fig. \ref{fig:5}(a), Fig. \ref{fig:5}(b), and Fig. \ref{fig:5}(c) respectively. For a fixed value of $\sigma$ each sequence is plotted in the same legend as $N$ is varied from $2^9$ to $2^{18}$. The rate $\frac{|\mathcal{I}|-|\mathcal{E}|}{N}$ is strictly increasing in $N$ (as one would expect \cite{sasoglu2009polarizationarbitrarydiscretememoryless}) and can be seen to approach the values of $I_{d,\mathrm{analog}}^{\mathrm{sq}}$. For each of these sequences as this rate is increasing, so is the upper bound on error decreasing showing that at this value of $\sigma$, we can provably suppress the error further using the asymptotic behavior of the Bhattacharya parameter.

We also note that the choice of kernel has a clear impact on performance. We find that for $d=5$, the kernel choice of $\alpha=2$ is far better than that of $\alpha=1$ in the sense that the polarization occurs more dramatically at the same value of $N$ as can be noted in Fig. \ref{fig:4}. Additionally due to the very large blocklengths we do not perform exact estimations of $P_{e,1}$ and $P_{e,2}$ and rather use the upper bound from Eq. \eqref{eq:polarineq}. It is worth noting that there aren't direct distance scaling guarantees for polar codes and in large part they aren't a direct measure of performance. Using the upper bound in Eq. \eqref{eq:polarineq}, the error probabilities converge to zero for large enough $N$ (provably upper bounded by $N2^{-N^\beta}$ for $0<\beta<1/2$ \cite{6357295}) which is essential in showing their capacity achieving properties. This also is sufficient to claim good performance without delving into distance since the error probability is ultimately the relevant metric for performance. In their application in classical settings, their distance can be improved by a careful selection of the information set \cite{8539599} if necessary, since the distance equals $2^{\min_{i\in\mathcal{A}}(wt(i-1))}$ where $wt(i-1)$ is the Hamming weight of the binary representation of the number $i-1$ for a classical polar code $(N,K,\mathcal{A})_d$. This can be seen to be true by simply considering that the encoding procedure of a polar code uses a binary tree structure and remains true for any kernel choice. For the kernel choice of $\alpha=1$ the smallest weight codeword can also just consist of $2^{\min_{i\in\mathcal{A}}(wt(i-1))}$ $1$s and the rest being zeroes. For a bit-flip noise like $W_1$ where usually $p(1,s_1)$ is much larger than any other $p_1(u,s_1)$ (except for $p(0,s_1)$) this shows how a different kernel choice can lead to the nearest non-zero codeword having smaller probability of being reached from the zero codeword which can then affect finite-blocklength performances.

\section{Achieving the capacity of pure-loss and amplification channels}
\begin{figure*}[ht]
    \centering
    \includegraphics[width=\linewidth]{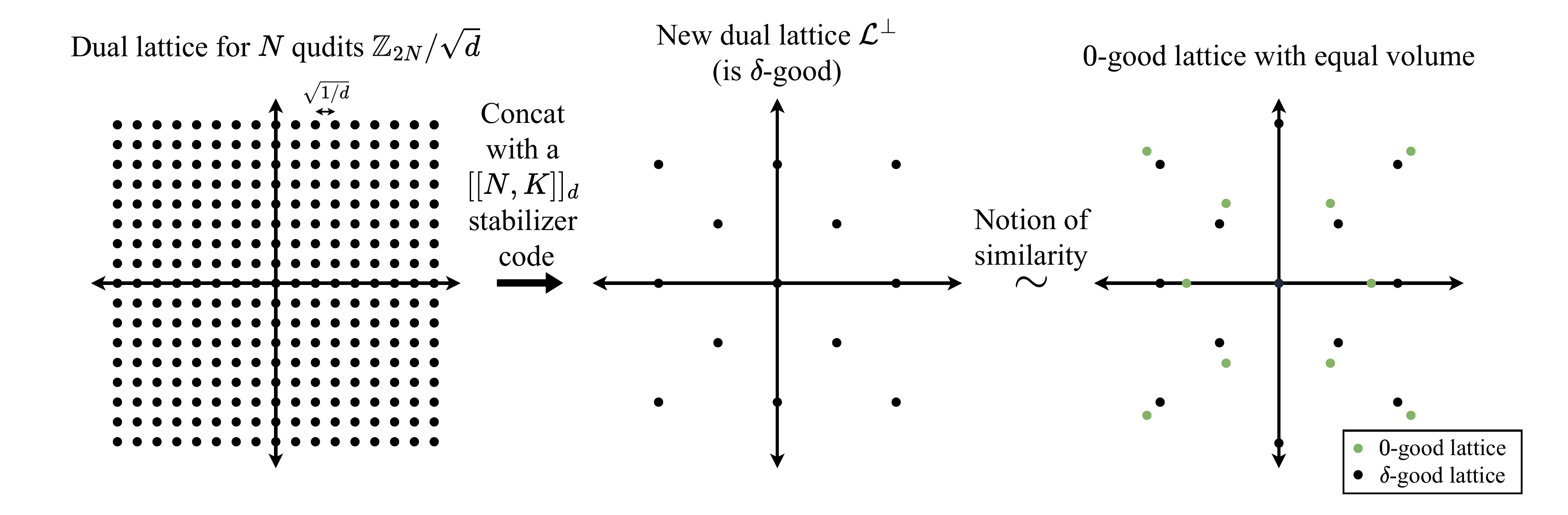}
    \caption{A rough visual depiction of the working principle behind how qudit stabilizer codes concatenated to square GKP qudits can be used to achieve the capacity of the pure-loss channel. Considering $N$ square GKP qudits by themselves, the dual lattice of this $N$ mode GKP code corresponds to the integer lattice scaled down by $\sqrt{d}$ which is $\mathbb{Z}^{2N}/\sqrt{d}$. Concatenation with some outer $[[N,K]]_d$ code restricts what can be considered a valid logical operation hence choosing some sub-lattice $\mathcal{L}^\perp\subseteq\mathbb{Z}^{2N}/\sqrt{d}$. A crucial aspect of increasing $d$ here shows that this allows for $\mathcal{L}^\perp$ to be chosen to have certain good properties. We note that this means that we can define some kind of notion of similarity between this concatenated lattice and another lattice of the same volume which is $0$-good (which has direct capacity achieving properties). We base this notion of similarity using $\delta$-good as a definition where as $\delta\to0$, we are able to capture the capacity achieving properties a $0$-good lattice would have. These lattices need not actually be similar in their exact descriptions but for the sake of illustration we depict them to be similar in the above figure. Using this $\delta$-good property, we are able to claim the existence of asymptotically good lattices to show that the capacity of pure-loss and the amplification channel is an achievable rate.}
    \label{fig:6}
\end{figure*}
\subsection{Quantum capacity of the pure-loss and amplification channels}\label{sec:caplossamp}
The bosonic pure-loss channel can be modeled as a beam-splitter interaction with the environment which is in the vacuum state following which the environment is traced out. Similarly the amplification channel can be modeled as a two-mode squeezing interaction with the environment in vacuum state following which the environment is traced out. We will consider the loss channel to have a transmittance $\eta<1$ and the gain of the amplification channel to be $G>1$. We can also write these noise channels as follows \cite{zheng2024performanceachievableratesgottesmankitaevpreskill}
\begin{equation}
    \mathcal{N}_{\mathrm{loss}}(\cdot) = \sum_{l=0}^{\infty} \hat{E}_l(\cdot)\hat{E}_l^\dagger,\quad \mathcal{N}_{\mathrm{amp}} = \sum_{l=0}^{\infty}\hat{A}_l(\cdot)\hat{A}_l^\dagger,
\end{equation}
where the Kraus operators are defined as
\begin{equation}
\begin{aligned}
    \hat{E}_l &= \left(\frac{\eta}{1-\eta}\right)^{l/2}\frac{\hat{a}^l}{\sqrt{l!}}(1-\gamma)^{\hat{n}/2},\\
    \hat{A}_l &= \frac{(G-1)^{l/2}}{\sqrt{l! G}}G^{-\hat{n}/2}(\hat{a}^{\dagger})^{l},
\end{aligned}
\end{equation}
and so it can be seen that they have a very similar form. Both of these channels are degradable \cite{PhysRevA.74.062307}, and so their capacity matches the maximal one-shot coherent information \cite{holevo1999evaluatingcapacitiesbosonicgaussian}. Hence their capacity is given by
\begin{equation}
    C_Q(\mathcal{N}_{\mathrm{loss}/\mathrm{amp}}) =\max\left(\log_2\left(\left|\frac{\tau}{1-\tau}\right|\right),0\right),
\end{equation}
where $\tau=\eta$ for loss and $\tau=G$ for amplification. In \cite{zheng2024performanceachievableratesgottesmankitaevpreskill}, due to the similar form of these two channels, it is shown that the infidelities of an infinite energy GKP code after experiencing this noise and then using the near-optimal transpose recovery are upper bounded by  $\frac{1}{4}\sum_{\pmb{x}\in\mathcal{L}^\perp\backslash{0}}e^{-\pi g|\pmb{x}|^2}$ (with $g=\frac{\eta}{1-\eta}$ for loss and $g=\frac{G}{G-1}$ for amplification) where $\mathcal{L}^\perp$ is the dual lattice of the multi-mode GKP code. This upper bound holds true for arbitrary multi-mode GKP codes, which also includes any concatenated square GKP code. Hence, if we find a capacity achieving sequence of codes for the pure-loss channel that make use of the transpose recovery, this same sequence of codes would also achieve the capacity of the amplification channel when using the transpose recovery. Hence for the rest of this discussion, we will show a capacity achieving sequence for the pure-loss since the extension to the amplification will trivially follow.

\subsection{Existence of a capacity achieving sequence of codes}\label{sec:lossresult}
While we have largely discussed square GKP qudits in this work, we now discuss general GKP lattices and proceed to show a capacity achieving sequence of lattices that are obtained through the concatenation of a square GKP qudit with a qudit stabilizer code obtained through a self-orthogonal code in $GF(d^2)$. We refer the reader to Appendix \ref{app:selfortho} for a quick introduction to quantum stabilizer codes obtained using self-orthogonal codes in $GF(d^2)$ where we mainly describe the ideas introduced in the works of \cite{959288,782103,1715533}. The set of all possible self-orthogonal $[N,(N-K)/2]_{d^2}$ codes will be denoted by $\mathcal{T}$. This set is useful since we can map each $[N,(N-K)/2]_{d^2}$ code to a valid stabilizer group which is a subset of the $N$ qudit ($d$ level) Pauli group generated by $N-K$ independent stabilizers. The self-orthogonality is the key feature which ensures that this is a valid stabilizer group. The set $\mathcal{T}$ has a notion of balanced-ness \cite{loeliger_averaging_1997} which is to say that the number of codes in $\mathcal{T}$ that contain a particular self-orthogonal $\pmb{a}\in\mathbb{F}^{N}_{d^2}\backslash\{0\}$ is independent of $\pmb{a}$ itself. Not all $\pmb{a}\in\mathbb{F}^{N}_{d^2}\backslash\{0\}$ are self orthogonal with $\frac{1}{d}(d^{2N}+(d-1)(-d)^N)-1$ elements of $\mathbb{F}^{N}_{d^2}\backslash\{0\}$ satisfying this property \cite{1362919}. This property is incredibly useful in proving existence of good codes by making claims on the average of some function taken over the set of codes $\mathcal{T}$. This is made use of in proving the Minkowski-Hlawka theorem for classical codes over prime fields \cite{loeliger_averaging_1997}, existence of good CSS codes \cite{PhysRevA.54.1098} and also existence of good stabilizer codes constructed from self orthogonal codes in $GF(d^2)$ \cite{959288}.

An $N$ mode GKP code is obtained through a stabilizer group generated using $2N$ independent displacements. Since all these displacements must necessarily commute, this enforces that the symplectic inner product of two displacements must be an integer multiple of $2\pi$. If we consider displacements defined as $\hat{D}(\pmb{\xi})=\exp(-\xi^T J \hat{\pmb{x}})$ where $\hat{x} = \begin{pmatrix}
    \hat{\pmb{q}}\\
    \hat{\pmb{p}}
\end{pmatrix}$ and $J$ is the symplectic form, we can define the stabilizer lattice $\mathcal{L}$ to be a symplectically integral lattice and the stabilizer group consists of $\hat{D}(\pmb{\xi})$ for any $\pmb{\xi}\in\mathcal{L}$. The connection of symplectically integral lattices to GKP codes was explored in \cite{PhysRevA.64.062301,gottesman_encoding_2001} and has been extensively studied in \cite{PRXQuantum.3.010335,Conrad2022gottesmankitaev} which also explore how gauge choices show up in the lattice formalism. For the lattice $\mathcal{L}$, the GKP code obtained encodes $\det(\mathcal{L})$ levels and the symplectic dual $\mathcal{L}^\perp$ contains all $\pmb{\xi}$ that have $D(\pmb{\xi})$ be a valid logical GKP operator. A valid GKP lattice must satisfy $\mathcal{L}\subseteq\mathcal{L}^\perp$. Additionally, we can rescale lattices to increase the number of levels it encodes since the lattice $\sqrt{\lambda}\mathcal{L}$ encodes $\lambda^N\det(\mathcal{L})$ levels. This is particularly useful when using self-dual lattices $\mathcal{L}^\perp=\mathcal{L}$ since they are unimodular $\det(\mathcal{L})=1$ and hence have to be rescaled by $\sqrt{\lambda}$ to encode information. However for the lattice to remain symplectically integral $\lambda$ must be an integer. The lattice of the GKP code corresponding to $N$ square qudits is $\sqrt{d}\mathbb{Z}^{2N}$ and so the dual of this is $\mathbb{Z}^{2N}/\sqrt{d}$. Concatenating some $[[N,K]]_d$ code to $N$ square GKP qudits restricts the set of what remains as a valid logical displacement, and so the new effective lattice of the overall $N$ mode code $\mathcal{L}^\perp$ is a sub-lattice of $\mathbb{Z}^{2N}/\sqrt{d}$ and $\det(\mathcal{L})=d^K$ which is the logical dimension of this lattice.

We now discuss the infidelity of a GKP code against bosonic pure-loss followed by a transpose recovery channel which is near optimal based on the results from \cite{zheng2024performanceachievableratesgottesmankitaevpreskill}. From Lemma 15 in Appendix F of \cite{zheng2024performanceachievableratesgottesmankitaevpreskill}, it follows that for any GKP code which is generated using a symplectically integral lattice $\mathcal{L}$, the infidelity after transpose recovery against pure-loss satisfies
\begin{equation}
    \epsilon\leq \frac{1}{4}\sum_{\pmb{x}\in\mathcal{L}^{\perp}\backslash\{0\}}e^{-\pi\frac{\eta}{1-\eta}|\pmb{x}|^2},
\end{equation}
which provides a very useful upper bound on the performance of a particular GKP code to the pure-loss channel. If we wish to show a rate to be achievable, it is sufficient to find a sequence of lattices $\mathcal{L}_n\subseteq\mathbb{R}^{2n}$ such that 
\begin{equation}
\begin{aligned}
    \lim_{n\to\infty} \left(\sum_{\pmb{x}\in\mathcal{L}_n^{\perp}\backslash\{0\}}e^{-\pi\frac{\eta}{1-\eta}|\pmb{x}|^2}\right)\to0 \text{ and }\\\frac{\log_2(\det(\mathcal{L}_n))}{n}\to R.
\end{aligned}
\end{equation}
In \cite{zheng2024performanceachievableratesgottesmankitaevpreskill}, it is shown if $2^R$ is an integer less than equal to $\frac{\eta}{1-\eta}$, it can be achieved through the rescaling of self-dual lattices ($\mathcal{L}=\mathcal{L}^\perp$). This exact limitation also showed up in the achievable rates for the Gaussian displacement noise \cite{PhysRevA.64.062301}, however as we have shown in this work, that can be surpassed through using concatenated square GKP qudits of large enough dimension. 
 Considering the family of unimodular lattices $\mathcal{L}=\mathcal{L}^\perp$, the Buser-Sarnak theorem \cite{buser1994period} showed that for any rotationally invariant $f:\mathbb{R}^{2N}\to\mathbb{R}$, the following holds
\begin{equation}
    \left\langle\sum_{\pmb{x}\in\mathcal{L}^{\perp}\backslash{\{0\}}}f(\pmb{x})\right\rangle = \int d^{2N}\pmb{x}f(\pmb{x})
\end{equation}
where the averaging over the set of all unimodular lattices satisfying $\mathcal{L}=\mathcal{L}^\perp$. This allows for the existence of a particular lattice which has the expression $\sum_{\pmb{x}\in\mathcal{\mathcal{L}^{\perp}}\backslash{\{0\}}}f(\pmb{x})\leq\det(\mathcal{L})\int d^{2N}\pmb{x}f(\pmb{x})$. This directly shows how much $\sqrt{\lambda}\mathcal{L}$ can be scaled up by while having the sum still upper bounded by something that approaches $0$ for $N\to\infty$. We will refer to such a lattice as a good lattice.

To demonstrate the capacity achieving properties for concatenated square GKP codes, we introduce the following loose definition.
\begin{definition}[$\delta$-good lattice]
    A symplectically integral lattice $\mathcal{L}$ which is weakly self-dual $\mathcal{L}\subseteq\mathcal{L}^\perp$ is $\delta$-good if it satisfies 
    \begin{equation}
    \sum_{\pmb{x}\in\mathcal{L}^{\perp}\backslash\{0\}}f(\pmb{x})\leq\det(\mathcal{L})\int d^{2N}\pmb{x}f(\pmb{x}) +\delta,
\end{equation}
for the $f(\pmb{x})=e^{-g\pi|\pmb{x}|^2}$ for any $g>0$.
\end{definition}
Hence now we only need to find a sequence of lattices that are $\delta$-good with $\delta\to0$ along this sequence. We find that under this notion of good-ness we can also define a similarity between a $0$-good lattice which is good in the same way that the capacity achieving self-dual lattices are to some $\delta$-good lattice obtained through concatenation

Consider the set $\mathcal{T}_{\mathcal{L}}$ containing the GKP lattices representing $[[N,K]]_d$ codes obtained by self-orthogonal $[N,(N-K)/2]_{d^2}$ codes concatenated to $N$ modes with each mode having square GKP qudit of $d$ levels. We show that over this set of lattices
\begin{equation}
    \left\langle\sum_{\pmb{x}\in\mathcal{L}^{\perp}\backslash{\{0\}}}f(\pmb{x})\right\rangle_{\mathcal{T}_\mathcal{L}} \leq d^K\int d^{2N}\pmb{x}f(\pmb{x}) + \delta_{N,d}
\end{equation}
for $f(\pmb{x})=e^{-\pi g|\pmb{x}|^2}$ for $g>0$ where $\delta_{N,d}\to0$ as $d\to\infty$ and $d\ln(d)\ll N\ll e^{\pi g d}$. This result makes use of the balanced-ness of the set $\mathcal{T}$ containing all the self-orthogonal $[N,(N-K)/2]_{d^2}$ codes \cite{959288,loeliger_averaging_1997}. Additionally increasing $d$ makes $\delta\to0$ since the balanced-ness of this set allows us to relate the average over $\mathcal{T}_\mathcal{L}$ to a sum over the lattice $\mathbb{Z}^{2N}/\sqrt{d}$. Note that any concatenated lattice is always a sub-lattice of $\mathbb{Z}^{2N}/\sqrt{d}$ and evaluating this sum over this lattice lets us directly relate it to the integral as a Riemann sum and the function $f(\pmb{x})=e^{-\pi g|\pmb{x}|^2}$ is particularly well-behaved for this to hold true. While these are not exact relations, we group all our errors into the term $\delta_{N,d}$ which we show approaches zero along this defined sequence of $N,d\to\infty$. Hence we now state our main result.
\begin{theorem}
    There exists a sequence of qudit stabilizer codes $[[N,K]]_d$ for $d$ being prime with $d=3\mod 4$, with increasing $N,d$ such that 
    \begin{equation}
        \log_2(d)\frac{K}{N} = \log_2\left(\frac{\eta}{1-\eta}\right) - \tilde{\epsilon}
    \end{equation}
    and $\tilde{\epsilon}$ can be made arbitrarily small simultaneously while the infidelity of this sequence of codes using transpose recovery after pure loss of transmittance $\eta$ converges to zero as $N,d\to \infty$ with $d\ln(d)\ll N\ll e^{\pi d\frac{\eta}{1-\eta}}$.
\end{theorem}
We refer the reader to Appendix \ref{app:loss} for a complete proof of the above result. 
We offer some more intuition for this result by noting that the increasing $d$ makes the lattice $\mathbb{Z}^{2N}/\sqrt{d}$ finer. Any particular $[[N,K]]_{d}$ code is essentially defining some choice of lattice $\mathcal{L}^\perp\subseteq\mathbb{Z}^{2N}/\sqrt{d}$ which clearly increases the available choice on increasing $d$. Further this means that beginning with a fine enough lattice $\mathbb{Z}^{2N}/\sqrt{d}$ lets us approximate a good lattice which is capable of good sphere packing. We motivate this idea in Fig. \ref{fig:6}. Hence we find that through the concatenation of the appropriate stabilizer code with square GKP qudits, the capacity of the loss channel is achieved as a rate for all values of transmission. While our proof is restricted to only prime values of $d$ with $d=3\mod 4$, we expect this to be generalizable for any $d$ which is a prime power and mainly use  $d=3\mod 4$ since the basis choice for $GF(d^2)$ allows for a relatively clearer mapping to the lattice formalism.

\section{Discussion}\label{sec:discuss}
In this work we have shown that the rates of coherent information for the Gaussian displacement noise channel, capacity of the pure-loss channel and the amplification channel are all achievable through the use of concatenated GKP square qudits by taking the limit of very large individual qudit dimension. While this limit may seem un-physical, we highlight that for reasonably small $d=5$ and $d=7$, we can already get incredibly close to the coherent information of the Gaussian displacement noise channel. In addition for the displacement noise channel, we are able to explicitly construct sequences of these codes through the use of quantum polar codes. This result hinges on the fact that the effective logical noise after correction for the displacement noise channel can be reduced into a classical mixture of qudit Pauli noise with an auxiliary system containing syndrome information. 

While our results have focused on infinite-energy GKP for the sake of finding achievable rates, we note that the method of using analog information can be extended to the finite energy case. Here we work under the twirling approximation used in previous works \cite{rozpkedek2021quantum,PhysRevA.101.012316} to treat finite-energy GKP as a Gaussian displacement noise channel on an infinite-energy GKP. If the innate noise model is purely Gaussian displacement noise with spread $\sigma_0$, making the data modes finite-energy just increases the effective noise to $\sqrt{\sigma_0^2+\sigma^2_{\mathrm{data}}}$ where $\sigma_{\mathrm{data}}^2=\tanh(\Delta^2/2)$. However, we note that if the ancillas are finite-energy, the syndrome information obtained in the ancillas are no longer faithful. This is since the ancillary modes itself have some random displacement noise associated to it. We can exactly quantify this uncertainty by accounting for it in a new probability distribution (assuming base-noise strength $\sigma$)
\begin{equation}
\begin{aligned}
    &p_{\sigma_{\mathrm{anc}}}(u,s_1) \\&= \frac{1}{2\sigma_{\mathrm{anc}}}\int d\xi e^{-\xi^2/4\sigma_{\mathrm{anc}}^2}p\left(u,\left(s_1 + \xi\frac{d}{2\pi}\right)\mod 1\right) \\&= \frac{1}{d}\theta_3\left(\pi\left(\frac{u+s_1}{d}\right),e^{-\pi(\sigma^2+2\sigma_{\mathrm{anc}}^2)/d}\right)
\end{aligned}
\end{equation}
where we have considered a teleportation based protocol for syndrome extraction \cite{PRXQuantum.3.010315} and $\theta_3(u,q) = \sum_{n\in\mathbb{Z}}q^{n^2}\cos(2nu)$ is a Jacobi-theta function. Note that this means that we can account for both finite-energy ancilla and data modes by increasing the effective noise $\sigma_{\mathrm{eff}} = \sqrt{\sigma_0^2+\sigma_{\mathrm{data}}^2 + 2\sigma_{\mathrm{anc}}^2}$. This shows a possible way of achieving the rate $I^{\mathrm{sq}}_{d,\mathrm{analog}}(\sigma_{\mathrm{eff}})$ if the base noise strength is $\sigma_0$. 

However, quantum polar codes have high weight stabilizers which poses a direct problem to making this a practical implementation using such codes. There are ways of fault tolerantly extracting syndrome information for quantum polar codes \cite{goswami2021quantum} and one could also use a Knill-Glancy scheme for the overall GKP code obtained through concatenation \cite{PhysRevA.105.042427} which would give $2N$ syndromes that are sufficient for decoding. However, these are resource heavy since they require preparation of an ancilla system which is encoded using the polar code, and would hence require a fault-tolerant encoding method for the outer polar code. Given the structure of the polar code encoder, we do not expect this to be an easy task on its own. Given the extensive amount of work done in constructing fault-tolerant methods of logical operations as well as syndrome extraction for QLDPC codes \cite{xu2024constant,xu2024fastparallelizablelogicalcomputation,williamson2024lowoverheadfaulttolerantquantumcomputation,swaroop2025universaladaptersquantumldpc,cowtan2025parallellogicalmeasurementsquantum,he2025extractorsqldpcarchitecturesefficient}, we note that a lot of these directly can be extended with our results for GKP codes while still achieving non-zero rates since we have only examined concatenated square GKP codes.

In regard to our results in the case of the pure-loss and amplification channels we would like to note that it is currently limited to an existence based result. Given that we find a capacity achieving sequence through qudit stabilizer codes, we believe that stronger claims on what these stabilizer codes can be possibly made by restricting to polynomial codes such as non-binary Reed-Muller codes \cite{KIM20083115,LaGuardia2012,golowich2024asymptoticallygoodquantumcodes} and leave this exploration for future work. Importantly, this result adds more understanding to the structure of what can be considered a good lattice since we are able to express it as a concatenated square lattice. These results further highlight the fundamental connection between lattices and error correcting codes \cite{ebeling2013lattices,loeliger_averaging_1997} and also offer an exploration of goodness for sphere-packing through our definition of a $\delta$-good lattice. We hope that our results spark more interest in exploration in regards to concatenated GKP codes.

\section*{Acknowledgments}
We would like to thank Kyungjoo Noh, Filip Rozp{\k{e}}dek, Gideon Lee, Debayan Bandopadhyay, Mariesa Teo, and Han Zheng for useful discussions. We acknowledge support from the ARO(W911NF-23-1-0077), ARO MURI (W911NF-21-1-0325), AFOSR MURI (FA9550-21-1-0209, FA9550-23-1-0338), DARPA (HR0011-24-9-0359, HR0011-24-9-0361), NSF (ERC-1941583, OMA-2137642, OSI-2326767, CCF-2312755, OSI-2426975), Packard Foundation (2020-71479), and the Marshall and Arlene Bennett Family Research Program. This material is based upon work supported by the U.S. Department of Energy, Office of Science, National Quantum Information Science Research Centers and Advanced Scientific Computing Research (ASCR) program under contract number DE-AC02-06CH11357 as part of the InterQnet quantum networking project.

\appendix

\section{A primer to polar codes}\label{app:polar}
\subsection{Encoding}\label{app:polarencoder}
The encoding scheme for polar codes \cite{5075875} can be understood as moving up a binary tree of depth $\log N$ to encode $N$ bits. The basic building block is shown in Fig. \ref{fig:7}. This is then used to construct a binary tree for the whole process. The complexity of encoding is $\mathcal{O}(N\log N)$ operations since each layer of the binary tree has a total of $N/2$ operations and there are $\log N$ layers. Here all addition is modulo $d$ where $d$ refers to the dimension of the dit.
\begin{figure}[ht]
  \centering
  \includegraphics[width=.6\linewidth]{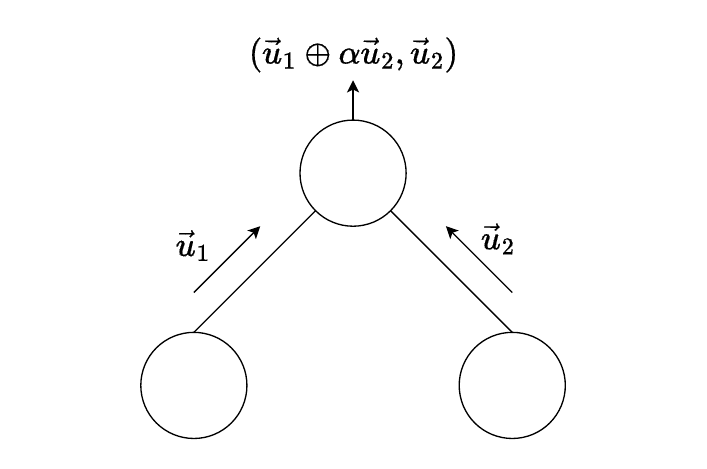}
  \caption{A unit of the encoding operation. The two leaves here take equal length dit-vectors and then do a ditwise xor between $\vec{u}_1$ and $\alpha\vec{u}_2$ while retaining a copy of $\vec{u}_2$ to output a dit-vector of twice the length.}\label{fig:7}
\end{figure}
\begin{figure}[ht]
  \centering
  \includegraphics[width=\linewidth]{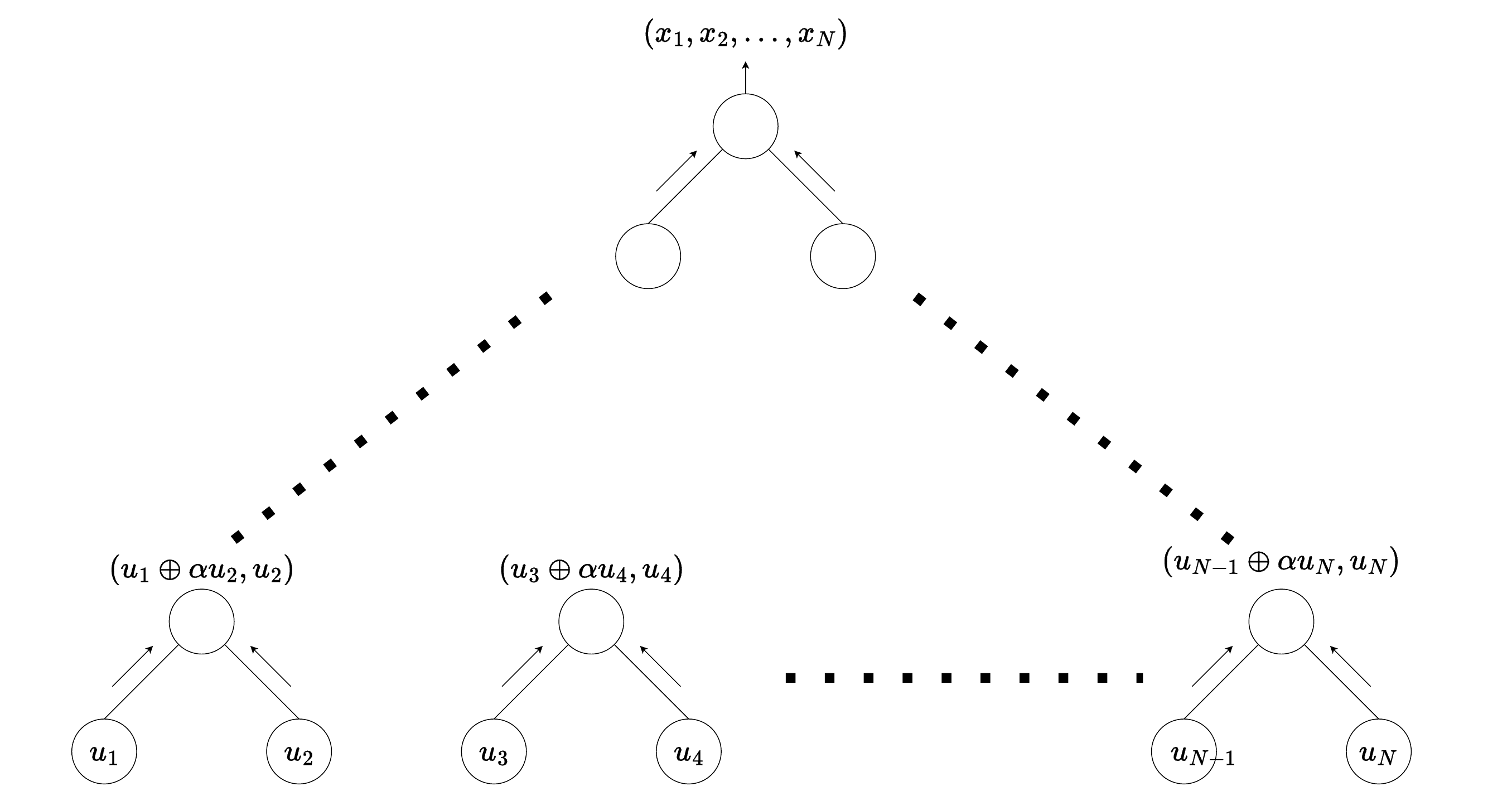}
  \caption{The full encoding binary tree for $N$ dits with the output $x^N_1 = u^N_1 G_N$}\label{fig:polartransform}
\end{figure} 
\begin{definition}[Polar transform]
    The polar transform over $N = 2^n$ dits $P_N:\mathbb{F}^{N}_d \to\mathbb{F}^{N}_d$ is defined by the action of the matrix $G_N$ on the $N$ bit vector $u^N_1$ as $P_N(u^N_1) = G_N u^N_1$ where all addition is modulo $d$ and
\begin{equation}
    G_N = \begin{pmatrix}
        1 & \alpha\\
        0 & 1
    \end{pmatrix}^{\otimes n}.
\end{equation}
The standard polar transform uses $\alpha=1$, but any choice $\alpha\in \mathbb{F}_d\backslash{0}$ can be used.
\end{definition}
\begin{definition}[Polar code $(N,K,\mathcal{A},u_{\mathcal{A}_c})_d$]
    A polar code $(N,K,\mathcal{A},u_{\mathcal{A}_c})_d$ is a code of block-length $N$ which encodes messages carrying $K = |\mathcal{A}|$ dits of information where $\mathcal{A}\subseteq \{1,\dots N\}$ and $u_{\mathcal{A}_c}\in \mathbb{F}^{N-K}_d$. The codewords for this code are obtained by encoding information in the bits with indices in set $\mathcal{A}$ and freezing all the other bits in $\mathcal{A}_c$ with the values $u_{\mathcal{A}_c}$ and then performing the polar transform on these input bits.
\end{definition}

\subsection{Successive cancellation decoding}\label{app:polarSCdecoder}
\begin{definition}[Discrete memoryless channel]
    We define a discrete memoryless channel (DMC) with input alphabet $\mathcal{X}$ and output alphabet $\mathcal{Y}$ where for any $y\in \mathcal{Y}$ and $x\in\mathcal{X}$, $W(y|x)$ equals the probability of obtaining output $y$ when the input was $x$. For any valid DMC, $\sum_{y\in\mathcal{Y}}W(y|x) = 1$. At any given time, the output of a DMC is only dependent on the input at that particular time, hence being memoryless.
\end{definition}
Here we describe the decoding procedure when each dit of the codewords experience the same DMC noise $W$. The decoding procedure makes use of the binary tree structure to achieve a complexity of $N\log(N)$ by making a decision of each bit only using the information of the output and the decisions of the bits which precede it.\\
To illustrate this idea, we first examine a single unit of the decoding procedure. Here we have $x_1 = u_1\oplus \alpha u_2$ and $x_2 = u_2$. We have given with use the following log-likelihood ratios
\begin{equation}
    \mathbf{L}_j = L^{[i]}_j = \log(\frac{p(x_j = 0)}{p(x_j = i)}),\quad i\in\{0,\dots d-1\}.
\end{equation}
Here $p(\mathrm{event})$ simply refers to the overall probability of said event occurring under a certain fixed circumstance which in this case would be corresponding to receiving noisy outputs $y_1$ and $y_2$ which correspond to inputs of $x_1$ and $x_2$ through noise channel $W$. We now wish to know the log-likelihood ratios for the bit $u_1$ agnostic to what the value of $u_2$ is. It then follows that
\begin{equation}
    p(u_1 = i) = \sum_{j=0}^{d-1}p(x_1 = i\oplus\alpha j)p(x_2 = j),
\end{equation}
Which then gives us
\begin{equation}
\begin{aligned}
    \log(\frac{p(u_1=0)}{p(u_1 = i)}) = f_\alpha(\mathbf{L}_1,\mathbf{L}_2)^{[i]},
\end{aligned}
\end{equation}
where we define the function $f_\alpha(\mathbf{L}_1,\mathbf{L}_2)$ taking two $d$ length vectors of LLRs and outputs one $d$ length vector of LLRs as
\begin{equation}
\begin{aligned}
    f_\alpha(\mathbf{L}_1,\mathbf{L}_2)^{[i]} = &\log(\sum_{j=0}^{d-1}\exp(-(L^{[\alpha j]}_1 + L^{[j]}_2))) \\&- \log(\sum_{j=0}^{d-1}\exp(-(L^{[\alpha j\oplus i]}_1 + L^{[j]}_2))).
\end{aligned}
\end{equation}
which we further generalize to taking inputs of two lists of $M$ different $d$ length vectors of LLRS defined as $\mathbf{L}^M_1$ and $\mathbf{K}^M_1$ respectively to be
\begin{equation}
\begin{aligned}
    f_\alpha(\mathbf{L}^M_1,\mathbf{K}_1^M)^{[i]}_m = &\log(\sum_{j=0}^{d-1}\exp(-(L^{[\alpha j]}_m + K^{[j]}_m))) \\&- \log(\sum_{j=0}^{d-1}\exp(-(L^{[\alpha j\oplus i]}_m + K^{[j]}_m))),
\end{aligned}
\end{equation}
where $f_\alpha(\mathbf{L}^M_1,\mathbf{K}_1^M)$ is also a list of $M$ different $d$ length vectors and the index $m$ in the above expression goes from $1$ to $M$. Now if we assume we have the exact knowledge of what $u_1$ is, we can use it to determine the probabilities for $u_2$ once this knowledge is included in the event
\begin{equation}
    p(u_2 = i|u_1) = p(x_1 = u_1\oplus \alpha i)p(x_2 = i).
\end{equation}
In terms of the LLRs, we have,
\begin{equation}
    \log(\frac{p(u_2 = 0|u_1)}{p(u_2 = i|u_1)}) = g_\alpha(\mathbf{L}_1,\mathbf{L}_2,u_1)^{[i]},
\end{equation}
where we define the function $g_\alpha(\mathbf{L}_1,\mathbf{L}_2,u_1)$ taking two $d$ length vectors of LLRs and one dit ($u$) and outputs one $d$ length vector of LLRs as
\begin{equation}
    g_\alpha(\mathbf{L}_1,\mathbf{L}_2,u)^{[i]} = L_2^{[i]} - L_1^{[u]} + L_1^{[u\oplus \alpha i]}.
\end{equation}
which we also further to taking inputs of two lists of $M$ different $d$ length vectors of LLRS defined as $\mathbf{L}^M_1$ and $\mathbf{K}^M_1$ respectively along with a list of $M$ dits $u^M_1$ as
\begin{equation}
    g_\alpha(\mathbf{L}_1^M,\mathbf{K}_1^M,u_1^M)^{[i]}_m = K^{[i]}_m - L^{[u_m]}_m + L_m^{[u_m\oplus \alpha i]}
\end{equation}
Once we have the LLRs for any particular bit $\mathbf{L}$, we make the guess of what the dit value is based on the function
\begin{equation}
    \texttt{dit}(\mathbf{L}) = \begin{cases}
        0 & \min_i L^{[i]} \geq 0\\
        \mathrm{argmin}_i L^{[i]} &\mathrm{otherwise}
    \end{cases}
\end{equation}
The above function selects the most likely value for the dit. Assuming we have knowledge of the outputs and the nature of the noise channel, we can calculate the input LLRs vector as $\mathbf{L}^N_1$ where
\begin{equation}
    L_i^{[j]} = \log(\frac{W(y_j|0)}{W(y_j|j)}),\quad i\in\{1,\dots N\},j\in\{0,\dots d-1\}
\end{equation}
The beliefs of the left child of the root in the encoding tree can be calculated to simply be $f_\alpha(\mathbf{L}_i,\mathbf{L}_{i+N/2})$ where $i=1$ to $N/2$. Suppose we now have a successive cancellation decoder for $N/2$ bits, we simply use the $N/2$ values of $f_\alpha(\mathbf{L}_i,\mathbf{L}_{i+N/2})$ as input to those to get guesses $\hat{u}^{N/2}_1$ which is the guessed output of the left child to the root node. We now can pass beliefs $g_\alpha(\mathbf{L}_i,\mathbf{L}_{i+N/2},\hat{u}_i)$ where $i=1$ to $N/2$ to the right child of the root and use the SC decoder for $N/2$ bits with these inputs and get a guess of what the outputs of the right child must have been which will finally give us the guess of the decoded codeword $\hat{x}^{N}_1$. The recursive nature of this is captured by the steps $L$, $R$ and then $U$ in Fig \ref{fig:9}.
\begin{figure*}[ht]
    \centering
    \includegraphics[width=\textwidth]{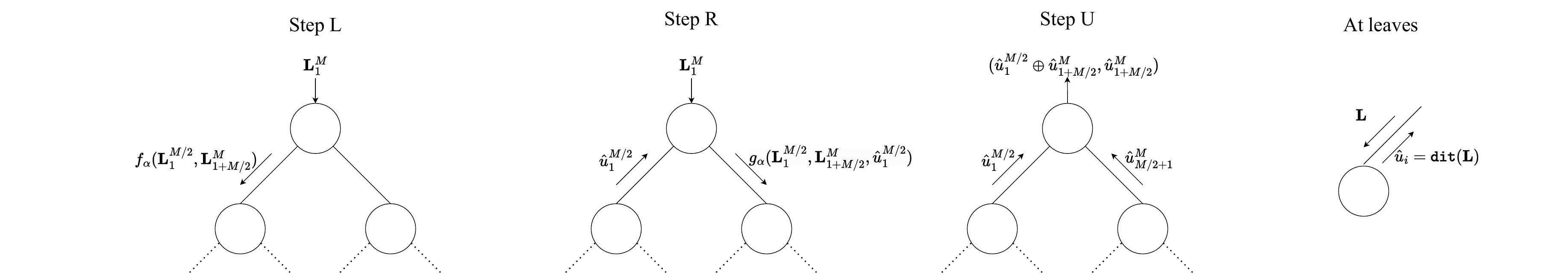}
    \caption{A unit for the decoding steps at each node. The steps proceed as step L (pass values to left child), step R (pass values to right child) and then step U (pass guessed dits to parent). If it reaches a leaf, the leaf returns a guess $\hat{u}_i$ using the function $\texttt{dit}$.}
    \label{fig:9}
\end{figure*}
The channel which dit $u_i$ experiences in this can be written as
\begin{equation}
\begin{aligned}
    &W^{(i)}((y^N_1,u^{i-1}_1)|u_i) \\& = \sum_{u^{N}_{i+1}\in \mathbb{F}^{N-i}_d}\frac{1}{d^{N-i}}\left[\prod_{j=1}^N W\left(y_j\bigg|(P_N(u^N_1))_j\right)\right]
\end{aligned}
\end{equation}
and the beliefs that reach the leaf nodes in the tree are simply the LLRs associated to $u_i$ with this channel.
\begin{theorem}[Classical polar coding]\label{thm:clPolar}
    For a discrete memory-less channel $W$ with input alphabet $\mathcal{X}$ with cardinality $d$ that is prime and an output alphabet $\mathcal{Y}$, there exists a sequence of polar code under successive cancellation decoding which can achieve coding rate $K/N$ arbitrarily close to 
    \begin{equation}
        I(W) = \sum_{x\in\mathcal{X}}\sum_{y\in\mathcal{Y}}\frac{1}{d}W(y|x)\log_d\left(\frac{W(y|x)}{\sum_{x'\in\mathcal{X}}\frac{1}{d}W(y|x')}\right)
    \end{equation}
    Here $W(y|x)$ is the probability of obtaining $y$ as a noisy output for input $x$.
\end{theorem}
\begin{proof}
    See reference \cite{sasoglu2009polarizationarbitrarydiscretememoryless}
\end{proof}
Note that while the analysis in \cite{sasoglu2009polarizationarbitrarydiscretememoryless} is done for discrete memoryless channels with a finite size of output alphabet, all their analysis can be extended to channels with a continuous output alphabet.

\section{Polar codes for classical mixtures of Pauli channels}\label{app:CMP}
\begin{definition}[Quantum polar transform]
We define the polar transform in the qudit $Z$ basis as follows
\begin{equation}
    \hat{V} = \sum_{\mathbf{z}\in \mathbb{F}^N_d}|G_N\mathbf{z}\rangle\langle \mathbf{z}|.
\end{equation}
Where all operations are performed modulo $d$ to stay in the field of $\mathbb{F}^N_d$. Using the standard definitions of the qudit $X$ and $Z$ bases being defined as the bases which diagonalize the following operators
\begin{equation}
    \hat{X} = \sum_{j=0}^{d-1}|j+1\rangle\langle j|,\quad \hat{Z} = \sum_{j=0}^{d-1}\omega^j|j\rangle\langle j|,
\end{equation}
we can rewrite the quantum polar transform in the qudit $x$ basis states $\ket{\tilde{\mathbf{x}}}$ by noting
\begin{equation}
\begin{aligned}
    V &=\sum_{\mathbf{z}\in \mathbb{F}^N_d}|G_N\mathbf{z}\rangle\langle \mathbf{z}|\\
    &=\sum_{\mathbf{z},\mathbf{x},\mathbf{x'}\in \mathbb{F}^N_d}\frac{e^{\frac{2\pi i}{d}(\mathbf{x'}\cdot G_N\mathbf{z} - \mathbf{x}\cdot\mathbf{}{z})}}{d}|\mathbf{\tilde{x}'}\rangle\langle \mathbf{\tilde{x}}|\\
    &= \sum_{\mathbf{x}\in \mathbb{F}^N_d}|G_N^{-T}\tilde{\mathbf{x}}\rangle\langle \tilde{\mathbf{x}}|.
\end{aligned}
\end{equation}
\end{definition}
In the last step, we have used the property that $\frac{e^{\frac{2\pi i}{d}(\mathbf{x'}\cdot G_N\mathbf{z} - \mathbf{x}\cdot\mathbf{z})}}{d}$ will sum to 1 if and only if $\mathbf{x'}\cdot G_N\mathbf{z} - \mathbf{x}\cdot\mathbf{z} = 0\mod d$ otherwise it sums to zero. Hence giving the condition $G_N^T\mathbf{x}'=\mathbf{x}$. This is in effect the same as saying that all CSUM gates in the $\mathbf{z}$ basis become into CSUM$^{-1}$ gates in the $\mathbf{x}$ basis with the control and target swapped. With the above definition, we can note that the polar transform acts simultaneously in the $X$ and $Z$ bases with the order of qudits reversed from one to the other and a different kernel since
\begin{equation}
    G_N^{-T} = \begin{pmatrix}
        1& 0\\
        -\alpha & 1
    \end{pmatrix}^{\otimes n}.
\end{equation}
However, as is pointed out in \cite{5513568}, this kernel also polarizes the channels in the same way.
\begin{theorem}[Quantum polar code achievable rate for Pauli noise]\label{thm:qPdeco}
    For a qudit ($d$ levels which is prime) Pauli noise described by 
    \begin{equation}
        \mathcal{N}(\cdot) = \sum_{u,v}p_{u,v} \hat{X}^u \hat{Z}^v(\cdot)(\hat{X}^u \hat{Z}^v)^\dagger,
    \end{equation}
    there exists a sequence of quantum polar codes using the SC decoder that asymptotically achieve the rate (in qudit per channel use)
    \begin{equation}
        R = I(W_A) + I(W_P) - 1,
    \end{equation}
    where 
    \begin{equation}
        W_A(z+u|z) = \sum_{v} p_{u,v} = p_u,\quad W_P((x+v,u)|x) = p_{u,v}
    \end{equation}
\end{theorem}
\begin{proof}
    We define for our convenience the following channels
    \begin{equation}
    \begin{aligned}
        W^c_A(z+u|z) &= W_A(z+d-u|z),\\ W^c_P((x+v,u)|x) &= W_P((x+d-v,u)|x).
        \end{aligned}
    \end{equation}
    Note that these channels are DMCs, and so will polarize for prime $d$, and represent a noise channel that acts on the amplitude ($z$) and phase ($x$) bases respectively hence some of the qudits will be good for amplitude and some will be good for phase. Let us consider two parties $A$ and $B$. $A$ has $N$ qudit bell pairs and sends $N$ of these halves through a noise channel to $B$. We can divide $N$ into 4 sets based on the channels $W^c_A$ and $W^c_P$
    \begin{itemize}
        \item $\mathcal{I}$ is the set of qudits which are good for both amplitude and phase
        \item $\mathcal{A}$ is the set of qudits bad for amplitude but good for phase
        \item $\mathcal{P}$ is the set of qudits bad for phase but good for amplitude
        \item $\mathcal{E}$ is the set of qudits bad for both amplitude and phase
    \end{itemize}
    We can define projectors $\Pi_{\mathcal{A}}$ which freezes amplitude values for set $\mathcal{A}$ and $\Pi_{\mathcal{P}}$ which freezes phase values for set $\mathcal{P}$. If $B$ decides to apply the polar transform on their qudits while the frozen values are decided beforehand by $A$ and $B$, we end up in the state
    \begin{equation}
        \ket{\psi_0} = \frac{1}{\sqrt{d^{N - |\mathcal{A}\cup\mathcal{P}|}}}\sum_{\mathbf{z}}\Pi_{\mathcal{A}}\Pi_{\mathcal{P}}\ket{\mathbf{z}}_A\otimes\ket{G_N\mathbf{z}}_B.
    \end{equation}
    Now $B$ suffers noise on their qudits (perhaps after some communication channel) which is described by the action of $\mathcal{N}$ on each qudit. We now consider an additional environment mode which tracks the error that has occurred hence ending up in the state
    \begin{equation}
        \ket{\psi_1} \propto \sum_{\mathbf{z},\mathbf{u},\mathbf{v}}\sqrt{p_{\mathbf{u},\mathbf{v}}}\Pi_{\mathcal{A}}\Pi_{\mathcal{P}}\ket{\mathbf{z}}_A\otimes X^{\mathbf{u}}Z^{\mathbf{v}}\ket{G_N\mathbf{z}}_B\otimes\ket{\mathbf{u},\mathbf{v}}_E.
    \end{equation}
    Applying the inverse of the polar transform defined by 
    \begin{equation}
        V^\dagger = \sum_{\mathbf{z}}|G_N^{-1}\mathbf{z}\rangle\langle\mathbf{z}| = \sum_{\mathbf{x}}|G_N^{-T}\tilde{\mathbf{x}}\rangle\langle\tilde{\mathbf{x}}|.
    \end{equation}
    We note the following
    \begin{gather}
        \hat{V}^\dagger \hat{X}^{\mathbf{u}}\hat{V} = \hat{X}^{\mathbf{u}'},\quad G_N\mathbf{u'} = \mathbf{u},\\
        \hat{V}^\dagger \hat{Z}^{\mathbf{v}}\hat{V} = \hat{Z}^{\mathbf{v}'},\quad G_N^{-T}\mathbf{v'} = \mathbf{v},
    \end{gather}
    which means that acting $\hat{V}^\dagger$ on $B$ gives
    \begin{equation}
        \ket{\psi_2} \propto\sum_{\mathbf{z},\mathbf{u},\mathbf{v}}\sqrt{p_{\mathbf{u},\mathbf{v}}}\Pi_{\mathcal{A}}\Pi_{\mathcal{P}}\ket{\mathbf{z}}_A\otimes \hat{X}^{\mathbf{u}'}\hat{Z}^{\mathbf{v}'}\ket{\mathbf{z}}_B\otimes\ket{\mathbf{u},\mathbf{v}}_E
    \end{equation}
    By measuring qudits in $\mathcal{A}\cup\mathcal{E}$ in the amplitude basis, we can obtain the dit values of $\mathbf{u}'$ corresponding to those. Further we know that $\mathbf{u}$ is nothing but the noisy output of the input $\mathbf{0}$ to the channel $W_A$ which is equivalent to treating $\mathbf{0}$ as the noisy output to the input $\mathbf{u}$ to the channel $W_A^c$ since $W_A^c(0|u) = W_A(u|0)$. However $\mathbf{u}$ is a codeword for the n-dit classical polar code with dits $|\mathcal{A}\cup\mathcal{E}|$ frozen to the values obtained on the measurement of those qudits in the amplitude basis. Hence the task of finding $\mathbf{u}'$ is simply the decoding of the classical polar code with noisy output of $\mathbf{0}$ from the channel $W_A^c$ which from theorem \ref{thm:clPolar} will achieve rate of $I(W_A^c)-\epsilon_1$ for $\epsilon_1>0$. Note that trivially $I(W_A^c) = I(W_A)$ and similarly $I(W_P^c) = I(W_P)$ since they only differ by a permutation. 
    
    After this decoding is complete we obtain some estimate of the amplitude error $\hat{\mathbf{u}}$ which on correcting, we end up in the state
    \begin{equation}
        \ket{\psi_3} \propto \sum_{\mathbf{x},\mathbf{v}}\sqrt{p_{\mathbf{v}|\hat{\mathbf{u}}}}\Pi_{\mathcal{A}}\Pi_{\mathcal{P}}\ket{\tilde{\mathbf{x}}}_A\otimes Z^{\mathbf{v}'}\ket{\tilde{\mathbf{x}}}_B\otimes\ket{\mathbf{\hat{u}},\mathbf{v}}_E + \mathcal{O}(\sqrt{p_A}),
    \end{equation}
    where $\sqrt{p_A}$ is the probability that the classical decoder for noise channel $W_A^c$ fails and as a result will contribute a small mixture of errors in the output state as a $\mathcal{O}(\sqrt{p_A})$ term. Similarly on measuring the qudits in set $\mathcal{P}\cup\mathcal{E}$ in the phase basis (note that we are able to do simultaneous phase basis and amplitude basis measurements for $\mathcal{E}$ since it consists of shared bell pairs) we obtain the dit values of $\mathbf{v}'$ corresponding to those. Similarly we can note that $(\mathbf{0},\hat{\mathbf{u}})$ is the noisy output of $\mathbf{v}$ (assuming the decoding of $\mathbf{u}$ is successful) through the noise channel $W_P^c$. This allows us to use the classical decoder again and this polar code achieves the rate of $I(W_P)-\epsilon_2$ for $\epsilon_2>0$. On decoding we will end up in the state
    \begin{equation}
    \begin{aligned}
        \ket{\psi_3} \propto\sum_{\mathbf{x},\mathbf{v}}\Pi_{\mathcal{A}}&\Pi_{\mathcal{P}}\ket{{\mathbf{z}}}_A\otimes \ket{{\mathbf{z}}}_B\otimes\ket{\hat{\mathbf{u}},\hat{\mathbf{v}}}_E\\ &+ \mathcal{O}(\sqrt{p_A}) + \mathcal{O}(\sqrt{p_P}).
    \end{aligned}
    \end{equation}
    We can make $p_A$ and $p_P$ arbitrarily small for any $\epsilon_1,\epsilon_2>0$. 
    
    Since all the sets $\mathcal{I},\mathcal{A},\mathcal{P}$ and $\mathcal{E}$ are disjoint, we have the following hold 
    \begin{gather}
        |\mathcal{A}| + |\mathcal{E}| = n(1 - I(W_A) + \epsilon_1),\\
        |\mathcal{P}| + |\mathcal{E}| = n(1 - I(W_P) + \epsilon_2),\\
        |\mathcal{I}| + |\mathcal{A}| + |\mathcal{P}| + |\mathcal{E}| = n.
    \end{gather}
    We can encode information in set $\mathcal{I}$ but also there are pre-shared entangled pairs which are $\mathcal{E}$. Hence the effective rate is given by
    \begin{equation}
        R = \frac{|\mathcal{I}|-|\mathcal{E}|}{n} = I(W_A) + I(W_P) - 1 - (\epsilon_1+\epsilon_2),
    \end{equation}
    where $\epsilon_1$ and $\epsilon_2$ can be made arbitrarily small as $n$ approaches infinity hence giving an asymptotic rate of $I(W_A) + I(W_P) - 1$. 
\end{proof}
\begin{corollary}
    The quantum polar code achieves the Hashing bound of a Pauli noise channel.
\end{corollary}
\begin{proof}
    The hashing bound for Pauli noise is defined as $1-H(\mathbf{p})$ where $\mathbf{p}$ is the probability vector associated to the Pauli noise $p_{u,v}$. We can note the following
    \begin{gather}
        I(W_P) = \sum_{u,v}p_{u,v}\log_d\left(\frac{p_{u,v}}{\frac{p_{u}}{d}}\right) = 1 - H(p_{u,v}) + H(p_{u}),\\
        I(W_A) = \sum_{u}p_u\log_d\left(\frac{p_{u}}{\frac{1}{d}}\right) = 1 - H(p_{u}).
    \end{gather}
    Hence giving
    \begin{equation}
        I(W_P) + I(W_A) - 1 = 1 - H(p_{u,v}) = 1 - H(\mathbf{p}),
    \end{equation}
    which shows that quantum polar codes achieve the Hashing bound.
\end{proof}
Now we can further generalize this idea to show that we can achieve the coherent information of a classical mixture of single qudit Pauli channels.
\begin{theorem}[Quantum polar code achievable rate for CMP]\label{thm:qPdeco}
    For a qudit ($d$ levels which is prime) Pauli noise described by 
    \begin{equation}
        \mathcal{N}(\cdot) = \sum_{u,v}p_{u,v,\mathbf{s}} \hat{X}^u \hat{Z}^v(\cdot)(\hat{X}^u \hat{Z}^v)^\dagger\otimes|\mathbf{s}\rangle\langle\mathbf{s}|
    \end{equation}
    there exists a sequence of quantum polar codes using the SC decoder that asymptotically achieve the rate (in per qudit)
    \begin{equation}
        R = I(W_A) + I(W_P) - 1
    \end{equation}
    where 
    \begin{equation}
    \begin{aligned}
        W_A((z+u,\mathbf{s})|z) &= \sum_{v} p_{u,v,\mathbf{s}} = p_u,\\ W_P((x+v,u,\mathbf{s})|x) &= p_{u,v,\mathbf{s}}
    \end{aligned}
    \end{equation}
\end{theorem}
\begin{proof}
    It is sufficient to show that there are classical polar codes which can achieve the rate $I(W_A)$ and $I(W_P)$ respectively. Note that these channels are simply memoryless channels with a finite sized input alphabet which is prime. It is known that the symmetric coherent information is achievable for these classical channels from \cite{sasoglu2009polarizationarbitrarydiscretememoryless}. Hence the rate $R$ can be achieved by constructing a CSS code with entanglement assistance.
\end{proof}
We believe that this can also be generalized to classical mixtures of more general multi-qudit Pauli noise, but would require a different kind of polar transforming kernel based on the Clifford channel combining principle introduced in \cite{9366784} and leave proving this for future work. 

\section{Lattice theta functions}
In this appendix we list a few results related to lattice theta functions and introduce relevant notation for the same. We define the following Jacobi-theta functions
\begin{equation}
    \theta_1(u,q) =2q^{1/4}\sum_{n=0}^{\infty}(-1)^nq^{n(n+1)}\sin((2n+1)u)
    \end{equation}
    \begin{equation}
        \theta_2(u,q) =2q^{1/4}\sum_{n=0}^{\infty}(-1)^nq^{n(n+1)}\cos((2n+1)u)\label{eq:theta2}
    \end{equation}
    \begin{equation}
    \theta_3(u,q) =1+2\sum_{n=1}^{\infty}q^{n^2}\cos(2nu)\label{eq:theta3} 
\end{equation}
which we will be making use of in the following appendices. When $\theta_3(u,q)$ is evaluated at $u=0$, we represent it as $\theta_3(q) = \theta_3(0,q)$. we define the lattice theta function for a lattice $\mathcal{L}\subseteq\mathbb{R}^{2N}$ as
\begin{equation}
    \Theta_{\mathcal{L}}(q)=\sum_{\pmb{x}\in\mathcal{L}}q^{|\pmb{x}|^2},
\end{equation}
from which it follows that $\Theta_\mathbb{Z^n}(q) = \theta_3(q)^n$. The theta function of the euclidean dual of $\mathcal{L}$, defined as $\mathcal{L}^{*}=\{\pmb{y}\in\mathbb{R}^{2N}|\pmb{y}\cdot\pmb{x}\in\mathbb{Z}\}$ is given by
\begin{equation}
    \Theta_{\mathcal{L}^{*}}(e^{i\pi z}) = \det(\mathcal{L})(i/z)^{N}\Theta_{\mathcal{L}}(e^{-i\pi/z}),\label{eq:dualtheta}
\end{equation}
where $z\in\mathbb{C}$. This relation is a special case of the Poisson summation formula shown in \cite{Elkies_2019} and was also noted in \cite{Conrad2022gottesmankitaev}. We consider $\mathcal{L}$ to be a symplectically integral lattice with the corresponding symplectic dual $\mathcal{L}^\perp$ (as defined in Eq. \eqref{eq:sympdual}). As noted in \cite{Conrad2022gottesmankitaev}, $\mathcal{L}^\perp$ is related to $\mathcal{L}^{*}$ by a rotational transform and hence they have the same lattice theta functions $\Theta_{\mathcal{L}^\perp}(q) = \Theta_{\mathcal{L}^{*}}(q)$. We note that the relation in Eq. \eqref{eq:dualtheta} also allows us to write
\begin{equation}
    \theta_3(e^{i\pi z}) = (i/z)\theta_3(e^{-i\pi/z}).
\end{equation}

\section{Pauli noise from closest lattice point GKP decoding on displacement noise}\label{app:CLPnoise}
Here we adapt the notation introduced in \cite{Conrad2022gottesmankitaev}. For convenience of the reader, we list some of the preliminary notation used in \cite{Conrad2022gottesmankitaev}, which we will also use. We have $N$ bosonic modes with quadratures $\hat{\pmb{q}}$ and $\hat{\pmb{p}}$ using which we define $\hat{\pmb{x}} = \begin{pmatrix}
    \hat{\pmb{q}}\\\hat{\pmb{p}}
\end{pmatrix}$ to represent a $2n$ dimensional phase space. The commutation relations can be written as
\begin{equation}
    [\hat{x}_k,\hat{x}_l] = iJ_{k,l},\quad J = \begin{pmatrix}
        0 & \mathbb{I}_n\\-\mathbb{I}_n & 0
    \end{pmatrix}.
\end{equation}
The matrix $J$ is the symplectic form and so we can define a displacement operator for a phase-space displacement $\pmb{\xi}\in\mathbb{R}^{2n}$ as
\begin{equation}
    \hat{D}(\pmb{\xi}) = \exp\left\{-i\sqrt{2\pi}\pmb{\xi}^T J \hat{\pmb{x}}\right\},
\end{equation}
which has the following properties
\begin{gather}
    \hat{D}^{\dagger}(\pmb{\xi})\hat{\pmb{x}}\hat{D}(\pmb{\xi}) = \hat{\pmb{x}} + \sqrt{2\pi}\pmb{\xi},\\
    \hat{D}(\pmb{\xi}_1)\hat{D}(\pmb{\xi}_2) = e^{-2\pi\pmb{\xi}_1^T J\pmb{\xi}_2}\hat{D}(\pmb{\xi}_2)\hat{D}(\pmb{\xi}_1),\\
    \hat{D}(\pmb{\xi}_1)\hat{D}(\pmb{\xi}_2) = e^{-\pi\pmb{\xi}_1^T J\pmb{\xi}_2}\hat{D}(\pmb{\xi}_1+\pmb{\xi}_2).
\end{gather}
For a GKP lattice over $n$ modes, we need exactly $2n$ displacements to generate the stabilizer group given by
\begin{equation}
    \mathcal{S} = \langle \hat{D}(\pmb{\xi}_1),\dots,\hat{D}(\pmb{\xi}_{2n})\rangle.
\end{equation}
These displacements must be linearly independent and must commute with each other, giving the condition $\pmb{\xi}_j^T J\pmb{\xi}_k\in\mathbb{Z}$ for all $j,k\in\{1,\dots 2n\}$. We now define the matrix
\begin{equation}
    M = \begin{pmatrix}
        \pmb{\xi}_1^T\\
        \pmb{\xi}_2^T\\
        \vdots\\
        \pmb{\xi}_{2n}^T
    \end{pmatrix},
\end{equation}
using which we define the symplectically integral lattice
\begin{equation}
    \mathcal{L} = \{\pmb{\xi}\in\mathbb{R}^{2n}\vert \pmb{\xi} = M^T \mathbf{a},\mathbf{a}\in\mathbb{Z}^{2n}\}.
\end{equation}
A GKP code defined using the following lattice would have its logical operations corresponding to displacements that are present in the dual lattice which would be defined by
\begin{equation}
    \mathcal{L}^\perp = \left\{\pmb{\xi}^\perp\in\mathbb{R}^{2n}\vert \left(\pmb{\xi}^\perp\right)^TJ\pmb{\xi}\in\mathbb{Z}\text{, }\forall\pmb{\xi}\in\mathcal{L}\right\}.\label{eq:sympdual}
\end{equation}
All the lattice points in $\mathcal{L}^{\perp}$ must necessarily be of form $\pmb{\xi}^{\perp} = (MJ)^{-1}\mathbf{b}$ where $\mathbf{b}\in\mathbb{Z}^{2n}$. The lattice $\mathcal{L}$ is self-dual which is to say that $\mathcal{L}\subseteq\mathcal{L}^\perp$.

Since each stabilizer is a displacement, one can measure the eigenvalues of these displacements to obtain stabilizers. If we consider a state $\ket{\psi}$ which is stabilized by $\mathcal{S}$, we have
\begin{equation}
    D(\pmb{\xi}_i)D(\mathbf{e})\ket{\psi} = e^{2\pi\pmb{\xi}_i^TJ\mathbf{e}}D(\mathbf{e})\ket{\psi}.
\end{equation}
which means we can measure the values $s_i = \pmb{\xi}_i^TJ\mathbf{e}$ modulo 1 without destroying the codeword information. These will be our syndrome measurements $\mathbf{s}$ each being modulo 1 and in the interval $\left[-\frac{1}{2},\frac{1}{2}\right)$ given by
\begin{equation}
    \mathbf{s} = (MJ \mathbf{e})\mod 1 = (MJ\mathbf{e}) - \lfloor MJ\mathbf{e}\rceil,
\end{equation}
for a particular value of $\mathbf{e}$ which we wish to correct. The initial guess for the correction is given by
\begin{equation}
    \pmb{\eta} = (MJ)^{-1}\mathbf{s},
\end{equation}
which would not account for possible logical errors. Using the closest-point decoding, one would find the closest lattice point in the dual lattice to this guessed correction. 

This gives the effective correction to be
\begin{equation}
    \bar{\pmb{\eta}} = \pmb{\eta} - \mathrm{argmin}_{\pmb{\xi}^\perp\in\mathcal{L}^\perp}\|\pmb{\eta} - \pmb{\xi}^{\perp}\|,
\end{equation}
for which the corresponding closest lattice point would be given by
\begin{equation}
    \mathbf{b}_0(\mathbf{s}) = \mathrm{argmin}_{\mathbf{b}\in\mathbb{Z}^{2n}}(\mathbf{s}-\mathbf{b})^T(MM^T)^{-1}(\mathbf{s}-\mathbf{b}).
\end{equation}
The logical error suffered by the correction of $\bar{\pmb{\eta}}$ would be 
\begin{equation}
    \mathbf{e} - \bar{\pmb{\eta}} = (MJ)^{-1}\left(\lfloor MJ\mathbf{e}\rceil + \mathbf{b}_0\right)
\end{equation}
Note that $\mathbf{e}-\bar{\pmb{\eta}}\in\mathcal{L}^\perp$ hence is necessarily a Pauli error. Note that we can obtain the Pauli group for this particular multi-mode GKP code as  $\mathcal{P} = \mathcal{L}^\perp/\mathcal{L}$ and for each particular $\mathbf{p}\in\mathcal{P}$, we can note that the probability of having occurred given that the syndrome is $\mathbf{s}$ would correspond to the occurrence of the event that $(MJ)^{-1}(\lfloor MJ\mathbf{e}\rceil + \mathbf{b}_0) = \mathbf{p}+\mathcal{L}$. This corresponds to having $\mathbf{e}$ be in the lattice defined by
\begin{equation}
    \mathcal{L}_{\mathbf{p},\mathbf{s}} = \mathbf{p} + (MJ)^{-1}(\mathbf{s} - \mathbf{b}_0) + \mathcal{L}
\end{equation}
which similarly can be extended to the lattice corresponding to a particular syndrome output as 
\begin{equation}
\mathcal{L}_{\mathbf{s}} = (MJ)^{-1}(\mathbf{s} - \mathbf{b}_0) + \mathcal{L}^{\perp}.
\end{equation}

Let us assume an underlying displacement noise channel of
\begin{equation}
    \mathcal{N}(\cdot) = \int d^2\pmb{\alpha} P(\pmb{\alpha}) \hat{D}(\pmb{\alpha})(\cdot) \hat{D}(\pmb{\alpha})^\dagger.
\end{equation}
We can now define the probability of doing a logical $\mathbf{p}$ Pauli operation and getting a syndrome of $\mathbf{s}$ corresponds to
\begin{equation}
    p(\mathbf{p},\mathbf{s}) \propto \sum_{\pmb{\xi}\in\mathcal{L}_{\mathbf{p},\mathbf{s}}}P(\pmb{\xi})
\end{equation}
and the conditional probability of having suffered some Pauli whilst knowing what the syndrome is would be
\begin{equation}
    p(\mathbf{p}|\mathbf{s}) = \frac{\sum_{\pmb{\xi}\in\mathcal{L}_{\mathbf{p},\mathbf{s}}}P(\pmb{\xi})}{\sum_{\pmb{\xi'}\in\mathcal{L}_{\mathbf{s}}}P(\pmb{\xi'})}
\end{equation}
Hence the effective logical noise channel after obtaining the syndrome $\mathbf{s}$ can be written as follows
\begin{equation}
    \mathcal{N}_{\mathrm{logical},\mathbf{s}}(\cdot) = \sum_{\mathbf{p}\in\mathcal{P}}p(\mathbf{p}|\mathbf{s}) L(\mathbf{p})(\cdot)L(\mathbf{p})^\dagger,
\end{equation}
where $L(\mathbf{p})$ is the logical Pauli operation corresponding to $\mathbf{p}$. 

The coherent information (in terms of qubit per channel use) of the above logical noise channel (taking input state to be the maximally mixed state) after having obtained the syndrome of $\mathbf{s}$ assuming that the Pauli group has dimension $d^2$ is given by
\begin{equation}
\begin{aligned}
    R_{\mathbf{s}} &= \log_2(d) + \sum_{\mathbf{p}} p(\mathbf{p}|\mathbf{s})\log_2(p(\mathbf{p}|\mathbf{s}))\\ &= \log_2(d) - H_2(p(\mathbf{p}|\mathbf{s})),
\end{aligned}
\end{equation}
which has a clear dependence on the syndrome. If we were to consider averaging this over the marginal distribution associated to achieving syndrome $\mathbf{s}$ we obtain
\begin{equation}
    R = \log_2(d) - H(\mathbf{p}|\mathbf{s}),
\end{equation}
where $H(\mathbf{p}|\mathbf{s})$ is the conditional entropy of $\mathbf{p}$ given $\mathbf{s}$. If one were to throw away the syndrome information, we are left with the marginal distribution $p(\mathbf{p})$. The information gain which can be achieved using the syndrome information is given by $I(\mathbf{p};\mathbf{s}) = H(\mathbf{p}) - H(\mathbf{p}|\mathbf{s})$ which we know will always be non-negative.

To show that this is always an achievable rate, we note that $\mathcal{N}_{\mathrm{logical}}$ can be understood as a classical mixture of Pauli channels written as
\begin{equation}
    \mathcal{N}_{\mathrm{logical}}(\cdot) = \int d\mathbf{s}\sum_{\mathbf{p}}p(\mathbf{p},\mathbf{s})L(\mathbf{p})(\cdot)L(\mathbf{p})\otimes|\mathbf{s}\rangle\langle\mathbf{s}|,
\end{equation}
for which the value $R$ is its coherent information (taking input state to be the maximally mixed state) which is hence an achievable rate. We also know that our specific construction of polar codes will achieve this rate in the case of a single qudit Pauli noise.
\subsection{Square GKP with Gaussian displacement noise}\label{app:dispSqGKP}
Let us consider the standard square GKP lattice defined by stabilizers
\begin{equation}
    \hat{S}_1 = \exp(i\hat{q}\sqrt{2\pi d}),\quad\hat{S}_2 = \exp(-i\hat{p}\sqrt{2\pi d}),
\end{equation}
and logical operators are $\hat{X}_L = S_2^{1/d}$ and $\hat{Z}_L = S_1^{1/d}$. We get stabilizer measurements $s_1 = \hat{q}\sqrt{\frac{d}{2\pi}} \mod 1$ and $s_2 = -\hat{p}\sqrt{\frac{d}{2\pi}} \mod 1$. Assuming a coherent displacement error to have occurred, the probability that closest point decoding results in a logical $X^u$ is given by
\begin{equation}
    \begin{aligned}\label{eq:pus1sup2}
    p_1(u,s_1) &= p(s_1)p(u|s_1)\\
    &=\frac{1}{\sigma\sqrt{d}}\sum_{l\in\mathbb{Z}}\exp(-\frac{\pi d}{\sigma^2}\left(\frac{s_1}{d}+\left(l+\frac{u}{d}\right)\right)^2)\\
    &= \frac{1}{d}\theta_3\left(\pi\left(\frac{u+s_1}{d}\right),e^{-\pi\sigma^2/d}\right),
\end{aligned}
\end{equation}
where $\theta_3$ is the Jacobi-Theta function defined in Eq. \eqref{eq:theta3} and the marginal distribution for 
\begin{equation}
    p(s_1) = \theta_3(\pi s_1,e^{-d\pi\sigma^2}).
\end{equation}
The joint distribution for $p_2(v,s_2) = p_1(v,-s_2)$ where $Z^v$ is the phase error suffered due to the coherent displacement followed by closest point decoding. We now have an achievable rate of
\begin{equation}
    I_{d,\mathrm{analog}}^{\mathrm{sq}} = \log_2(d) + 2\int ds_1 \sum_{u=0}^{d-1} p_1(u,s_1)\log_2\left(\frac{p_1(u,s_1)}{p(s_1)}\right).
\end{equation}
\begin{theorem}
    The rate $I_{d,\mathrm{analog}}$ is achievable by concatenation of GKP code encoding a qudit of $d$ levels (where $d$ is prime) with a quantum polar code designed for the noise channel $W_1$ on both amplitude and phase.
\end{theorem}
\begin{proof}
    We first note that the Pauli errors corresponding to $k_1$ and $k_2$ are independent of each other hence both can be treated to experience the noise channel described by $W_1$. We begin with the same setup of freezing certain bits based on the sets $\mathcal{I},\mathcal{A},\mathcal{P}$ and $\mathcal{E}$. The first step we do is correct each individual GKP mode. This results in giving $\mathbf{k}_1$ and $\mathbf{k}_2$ as vectors of the stabilizer values and the state ending up in
    \begin{equation}
    \ket{\Psi_1} = \sum_{\mathbf{z},\mathbf{u},\mathbf{v}}p_{\mathbf{u},\mathbf{v}|\mathbf{k}_1,\mathbf{k}_2} \ket{\mathbf{z}}_A X^{\mathbf{u}}Z^{\mathbf{v}}\ket{G_N\mathbf{z},\mathbf{k}_1,\mathbf{k}_2}_B\otimes\ket{\mathbf{u},\mathbf{v}}_E,
\end{equation}
Now we can perform the decoding procedure the same way as in the proof of \ref{thm:qPdeco} whilst also using the outputs of $\mathbf{k}_1$ and $\mathbf{k}_2$ in the decoding procedure which essentially achieves the capacity for the classical channel $W_1$ as well as $W_2$ hence achieving $I_{d,\mathrm{analog}} = 2I(W_1)-1$ asymptotically (note that $I(W_1) = I(W_2)$ by the fact that $p_1(u,s_1) = p_2(u,-s_1)$).
\end{proof}

\begin{theorem}
    For any given $\sigma$, there exists a $d_0$ such that for all prime $d>d_0$, the expression
    \begin{equation}
        \left|I^{\mathrm{sq}}_{d,\mathrm{analog}} - \log_2\left(\frac{1}{\sigma^2 e}\right)\right|
    \end{equation}
    can be made arbitrarily close to 0.
\end{theorem}
Since we know that $d$ is prime and a large number we know that it is an odd number and so we can consider that the variable $u$ can take on values from $-\lfloor d/2\rfloor$ to $\lfloor d/2\rfloor$ in steps of $1$. For our convenience, we rewrite
\begin{equation}
\begin{aligned}
    I_{d,\mathrm{analog}}& = \log_2\left(\frac{1}{\sigma^2}\right) \\&+ 2\int_{-1/2}^{1/2}ds_1\left(\sum_u p_1(u,s_1)\log( \frac{p_1(u,s_1)\sigma \sqrt{d}}{p(s_1)})\right).
\end{aligned}
\end{equation}

In the limit of $d\gg\frac{1}{\sigma^2}$, we can note that 
\begin{equation}
p(s_1) = 1 + \mathcal{O}(e^{-d\pi\sigma^2})    
\end{equation}
since each syndrome value between $-1/2$ and $1/2$ becomes nearly equally likely. This follows from the fact that for a given syndrome, the set of displacements that give the same syndrome are spaced apart by $\sqrt{2\pi/d}$ which gets increasingly fine as $d$ increases. This nature can also be noted by the fact that $p(s_1)$ is a theta function.

In the summation of $p_1(u,s_1)$, the contribution from terms of $l=0$ is significantly higher than any other term in the limit $d/\sigma^2\gg1$. Hence we can approximate the following truncated distribution of
\begin{equation}
    p_{\lim}(u,s_1) = \frac{1}{\sigma\sqrt{d}}\exp(-\frac{\pi d}{\sigma^2}\left(\frac{s_1}{d} + \frac{u}{d}\right)^2).
\end{equation}
To understand how far this might be from the actual probability distribution, we can integrate it to obtain
\begin{align}
    &\int_{-1/2}^{1/2}ds_1 \sum_u p_{\lim}(u,s_1) = \frac{1}{2}\left(\mathrm{erf}\left(\frac{\sqrt{d\pi}}{2\sigma}\right)\right)\\
    &= 1 - e^{-\frac{d\pi}{4\sigma^2}}\left(\frac{2\sigma}{\pi\sqrt{d}}+\mathcal{O}\left(\left(\sigma^2/d\right)^{-3/2}\right)\right),
\end{align}
which shows that the probability distribution is very close to $p_{\mathrm{lim}}(u,s_1)$ when we have $e^{-\frac{d\pi}{4\sigma^2}}\to 0$. To accurately denote the differences in using this function instead of $p_1(u,s_1)$ we can write $p_1(u,s_1) = p_{\mathrm{lim}}(u,s_1) + c(u,s_1)$ where $c(u,s_1)$ is appropriately defined for this. To upper bound $c(u,s_1)$ we note that it is comprised of the terms in equation \ref{eq:pus1} and if we individually maximize each term in the summation of $l$ excluding $l=0$ we can note that it corresponds to summing up $\frac{2}{\sigma\sqrt{d}}\exp(-\pi d (l-1/2)^2/\sigma^2)$ from $l=1$ to $\infty$ giving
\begin{equation}
    |p_1(u,s_1)-p_{\mathrm{lim}}(u,s_1)| = |c(u,s_1)| \leq \frac{\theta_2(0,e^{-d\pi/\sigma^2})}{\sigma\sqrt{d}}.
\end{equation}
Note that $\theta_2(0,q) \leq 2q^{1/4}(1-q)^{-1}$ when $q<1$, and since we are working in the limit of small $e^{-\pi d/\sigma^2}$, we can hence note that $\sigma\sqrt{d}|c(u,s_1)| = \mathcal{O}(e^{-\frac{\pi d}{4\sigma^2}})$ as illustrated in Fig. \ref{fig:1}(c). 

For $d\gg\max\{\frac{1}{\sigma^2},\sigma^2\}$, we expand $I^{\mathrm{sq}}_{d,\mathrm{analog}}$ as
\begin{widetext}
    \begin{equation}
    I^{\mathrm{sq}}_{d,\mathrm{analog}} = \log_2\left(\frac{1}{\sigma^2}\right) + 2\int_{-1/2}^{1/2} ds_1\left(\left(\sum_{u}(p_{\mathrm{lim}}(u,s_1) + c(u,s_1))\log(\sigma\sqrt{d}(p_{\mathrm{lim}}(u,s_1) + c(u,s_1)))\right) - p(s_1)\log(p(s_1))\right)
\end{equation}
\end{widetext}
We first note that $p(s_1)\log(p(s_1)) = \mathcal{O}(e^{-d\pi\sigma^2})$ since $p(s_1) = 1 + \mathcal{O}(e^{-d\pi\sigma^2})$. On further expansion we get
\begin{equation}
\begin{aligned}
    &I^{\mathrm{sq}}_{d,\mathrm{analog}} = \log_2\left(\frac{1}{\sigma^2}\right) \\&+ 2\log_2(e)(w_1+w_2+w_3+w_4)+ \mathcal{O}(e^{-d\pi\sigma^2}),
\end{aligned}
\end{equation}
where
\begin{align}
    w_1 &= \sum_{u}\int_{-1/2}^{1/2} ds_1 p_{\mathrm{lim}}(u,s_1)\ln(\sigma\sqrt{d}p_{\mathrm{lim}}(u,s_1)),\\
    w_2 &= \sum_{u}\int_{-1/2}^{1/2} ds_1 c(u,s_1)\ln(\sigma\sqrt{d}p_{\mathrm{lim}}(u,s_1)), \\
    w_3 &= \sum_{u}\int_{-1/2}^{1/2} ds_1 p_{\mathrm{lim}}(u,s_1)\ln(1 + \frac{c(u,s_1)}{p_{\mathrm{lim}}(u,s_1)}), \\
    w_4 &= \sum_{u}\int_{-1/2}^{1/2} c(u,s_1)\ln(1 + \frac{c(u,s_1)}{p_{\mathrm{lim}}(u,s_1)}) ds_1,
\end{align}
and we will now proceed to show that $w_1$ gives a leading order term in the limit of large $d$ while none of the other terms do.\\
\textbf{Evaluating $w_1$:} We first evaluate the integral which is fairly easy since it is a Gaussian definite integral which simplifies to
\begin{equation}
    \begin{aligned}
    w_1 &= \sum_u \frac{(e^{-\frac{\pi(1+2u)^2}{4d\sigma^2}}(2u+1)-e^{-\frac{\pi(2u-1)^2}{4d\sigma^2}}(2u-1))}{4\sigma\sqrt{d}}\\&+ \sum_u\frac{1}{2}\left(\mathrm{erf}\left(\frac{\sqrt{\pi}(1-2u/d)}{2\sigma\sqrt{d}}\right) + \mathrm{erf}\left(\frac{\sqrt{\pi}(1+2u/d)}{2\sigma\sqrt{d}}\right)\right)\\
    &= \frac{1}{2}\mathrm{erf}\left(\frac{\sqrt{d\pi}}{2\sigma}\right) + \frac{1}{2}\frac{e^{-\frac{\pi d}{4\sigma^2}}\sqrt{d}}{\sigma} \\&= \frac{1}{2} + \exp(-\frac{\pi d}{4\sigma^2})\frac{\sqrt{d}}{\sigma}\left(\frac{4-\pi}{2\pi}\right) \\&= \frac{1}{2} + \mathcal{O}(d^{1/2}e^{-\frac{d\pi}{4\sigma^2}}\sigma^{-1}).
\end{aligned}
\end{equation}
\textbf{Evaluating $w_2$:} We note that $|\log(\sigma\sqrt{d}p_{\mathrm{lim}}(u,s_1))| = \frac{\pi d}{\sigma^2}(\frac{s_1}{d}-\frac{u}{d})^2 \leq \frac{\pi d}{\sigma^2}$. Since $c(u,s_1)$ is positive, this means $w_2 < 0$. We can bound it as
\begin{align}
    |w_2| &\leq \sum_{u}\int_{-1/2}^{1/2} ds_1 c(u,s_1)\frac{\pi d}{\sigma^2}\\ &= \frac{\pi d}{\sigma^2} \left(1 - \sum_{u}\int_{-1/2}^{1/2} ds_1 p_{\mathrm{lim}}(u,s_1)\right)\\ 
    &= \frac{2\sqrt{d}}{\sigma}e^{-\frac{d\pi}{4\sigma^2}}\left(1 + \mathcal{O}\left(\left(\sigma^2/d\right)^{-3/2}\right)\right),
\end{align}
hence giving $|w_2| = \mathcal{O}\left(d^{1/2}e^{-\frac{d\pi}{4\sigma^2}}\sigma^{-1}\right)$.\\
\textbf{Evaluating $w_3$:} We can first note that $w_3$ is positive in value. The minimum value of $\sigma\sqrt{d}p_{\mathrm{lim}}(u,s_1)$ over the range of values that $u$ and $s_1$ take is $e^{-\frac{d\pi}{4\sigma^2}}$ when $dx = \lfloor\frac{d}{2}\rfloor$ (or $-\lfloor\frac{d}{2}\rfloor$) and $s_1 = - 1/2$ (or $1/2$). We can break the summation in $w_3$ into three intervals, $-\lfloor\frac{d}{2}\rfloor\leq u < -\lfloor\frac{d}{3}\rfloor$, $-\lfloor\frac{d}{3}\rfloor\leq u \leq\lfloor\frac{d}{3}\rfloor$ and $\lfloor\frac{d}{3}\rfloor<u\leq \lfloor\frac{d}{2}\rfloor$. When in the intervals  $-\lfloor\frac{d}{3}\rfloor\leq u \leq\lfloor\frac{d}{3}\rfloor$ and $\lfloor\frac{d}{3}\rfloor<u\leq \lfloor\frac{d}{2}\rfloor$, $e^{-\frac{\pi d}{4\sigma^2}}\leq \sigma\sqrt{d}p_{\mathrm{lim}}(u,s_1) \leq \alpha_1 e^{-\frac{\pi d}{9\sigma^2}}$ where $\alpha_1\approx 1$. Hence in these intervals
\begin{equation}
\begin{aligned}
    &p_{\mathrm{lim}}(u,s_1)\log(1 + \frac{c(u,s_1)}{p_{\mathrm{lim}}(u,s_1)})\\ &\leq\frac{1}{\sigma\sqrt{d}}\alpha_1e^{-\frac{\pi d}{9\sigma^2}}\log(1+ \frac{\mathcal{O}(e^{-\frac{d\pi}{4\sigma^2}})}{e^{-\frac{d\pi}{4\sigma^2}}}) = \mathcal{O}\left(\frac{e^{-\frac{\pi d}{9\sigma^2}}}{\sigma\sqrt{d}}\right),
\end{aligned}
\end{equation}
and in the interval of $-\lfloor\frac{d}{3}\rfloor\leq u \leq\lfloor\frac{d}{3}\rfloor$, we have $\sigma\sqrt{d}p_{\mathrm{lim}}(u,s_1) \geq \alpha_2 e^{-\frac{\pi d}{9\sigma^2}}$ where 
\begin{equation}
\begin{aligned}
    &p_{\mathrm{lim}}(u,s_1)\log(1 + \frac{c(u,s_1)}{p_{\mathrm{lim}}(u,s_1)}) \\&\leq p_{\mathrm{lim}}(u,s_1)\log(1+ \frac{\mathcal{O}(e^{-\frac{d\pi}{4\sigma^2}})}{\alpha_2e^{-\frac{d\pi}{9\sigma^2}}})\\ &= p_{\mathrm{lim}}(u,s_1)\log(1+ \mathcal{O}(e^{-\frac{5d\pi}{36\sigma^2}})),
\end{aligned}
\end{equation}
and since we know that $e^{-\frac{d\pi}{\sigma^2}}$ is a small number this gives us $p_{\mathrm{lim}}(u,s_1)\log(1 + \frac{c(u,s_1)}{p_{\mathrm{lim}}(u,s_1)}) = \mathcal{O}(e^{-\frac{5d\pi}{36\sigma^2}})p_{\mathrm{lim}}(u,s_1)$. Summing over all these intervals, we get
\begin{equation}
\begin{aligned}
    |w_3| &\leq \sum_{|u|\leq\lfloor\frac{d}{3}\rfloor}\mathcal{O}(e^{-\frac{5d\pi}{36\sigma^2}})p_{\mathrm{lim}}(u,s_1)\\ &\quad\quad\quad\quad\quad+ \sum_{\lfloor\frac{d}{3}\rfloor<|u|\leq\lfloor\frac{d}{2}\rfloor}\mathcal{O}\left(\frac{e^{-\frac{\pi d}{9\sigma^2}}}{\sigma\sqrt{d}}\right) \\&\leq \mathcal{O}(e^{-\frac{5d\pi}{36\sigma^2}}) + \mathcal{O}\left(d^{1/2}e^{-\frac{\pi d}{9\sigma^2}}\sigma^{-1}\right).
\end{aligned}
\end{equation}
\textbf{Evaluating $w_4$:} as is already shown before, $\sigma\sqrt{d}p_{\mathrm{lim}}(u,s_1)\geq e^{-\frac{d\pi}{4\sigma^2}}$. Using this we can place a bound on $w_4$ as
\begin{equation}
    |w_4| \leq \alpha_3\sum_{u}\int_{-1/2}^{1/2}ds_1 c(u,s_1) = \mathcal{O}(\sigma d^{-1/2}e^{-\frac{d\pi}{4\sigma^2}}).
\end{equation}
Now we put together all of these together to get
\begin{align}
    \left|I^{\mathrm{sq}}_{d,\mathrm{analog}} - \log_2\left(\frac{1}{\sigma^2e}\right)\right| \leq \mathcal{O}( d^{1/2}e^{-\frac{d\pi}{9\sigma^2}}\sigma^{-1})+ \mathcal{O}(e^{-d\pi\sigma^2}),
\end{align}
where this upper bound has a prefactor $\alpha_2$ such that in the interval of $-\lfloor\frac{d}{3}\rfloor\leq u \leq\lfloor\frac{d}{3}\rfloor$, we have $\sigma\sqrt{d}p_{\mathrm{lim}}(u,s_1) \geq \alpha_2 e^{-\frac{\pi d}{9\sigma^2}}$. The existence of fixed $\alpha_2$ (with respect to $d$) holds true as long as we have $d>d_0$ such that $\frac{d_0}{\sigma^2}\gg1$. As such this upper bound will also tend to zero hence proving our claim. Therefore we have shown that
\begin{equation}
    I^{\mathrm{sq}}_{d,\mathrm{analog}} = \log_2\left(\frac{1}{\sigma^2e}\right) -\mathcal{O}( d^{1/2}e^{-\frac{d\pi}{9\sigma^2}}\sigma^{-1})- \mathcal{O}(e^{-d\pi\sigma^2}).
\end{equation}
\subsection{Rectangular GKP with Gaussian displacement noise}
\begin{figure*}[ht]
    \centering
    \includegraphics[width=0.9\linewidth]{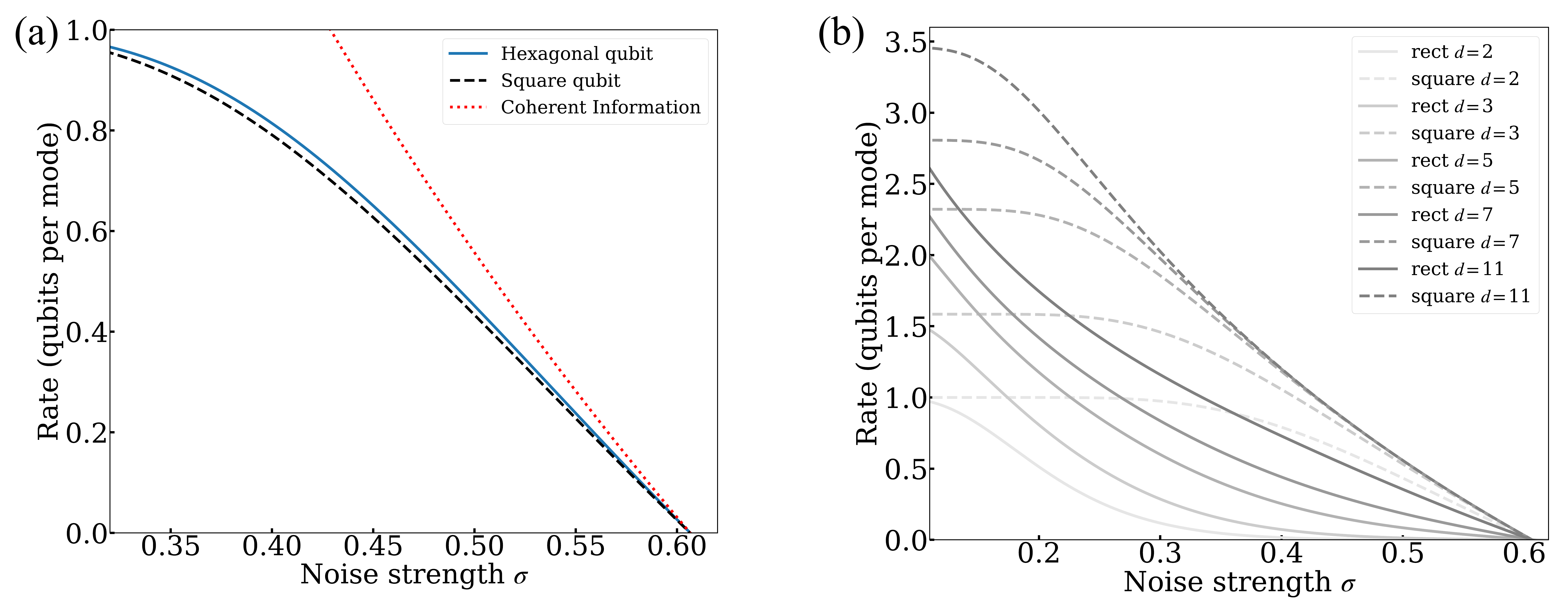}
    \caption{(a) The achievable rates for rectangular GKP qudits with $f=3$. (b) The achievable rates for a GKP hexagonal qubit slightly exceed that of a square GKP qubit as seen in this plot.}
    \label{fig:10}
\end{figure*}
We now examine the case of rectangular GKP with Gaussian displacement noise which would also result in the distributions associated to bit-flip and phase-flips be independent. Without loss of generality, let us assume the rectangle to be squeezed more along the $\hat{p}$ quadrature which would introduce a bias toward decreasing bit-flip errors and increasing phase-flip errors. The stabilizers are displacements defined as follows
\begin{equation}
    S_1 = \exp(i\hat{q}\frac{\sqrt{2\pi d}}{f}),\quad S_2 = \exp(-i\hat{p}f\sqrt{2\pi d}),
\end{equation}
where logical operators $X_L = S_2^{1/d}$ and $Z_L = S_1^{1/d}$ and stabilizer measurements $s_1 = \hat{q}f\sqrt{\frac{d}{2\pi}} \mod 1$ and $s_2 = -\hat{p}\frac{1}{f}\sqrt{\frac{d}{2\pi}} \mod 1$ where without loss of generality we assume $f>1$. Assuming a coherent displacement error to have occurred, the probability that closest point decoding results in a logical $X^u$ is given by
\begin{equation}
    \begin{aligned}
    &p_1(u,s_1)\\ &= \frac{1}{f\sigma\sqrt{d}}\sum_{l\in\mathbb{Z}}\exp(-f^2\frac{\left(s_1\sqrt{\frac{2\pi}{d}}+(dl+u)\sqrt{2\pi/d}\right)^2}{2f^2\sigma^2})\\
    &= \frac{1}{d}\theta_3\left(\pi\left(\frac{u+s_1}{d}\right),e^{-\frac{\pi\sigma^2}{f^2d}}\right),
\end{aligned}
\end{equation}
and similarly the probability of a logical $Z^v$ is given by
\begin{equation}
    \begin{aligned}
    &p_2(v,s_2)\\ &= \frac{f}{\sigma\sqrt{d}}\sum_{l\in\mathbb{Z}}\exp(-\frac{1}{f^2}\frac{\left(s_2\sqrt{\frac{2\pi}{d}}-(dl+v)\sqrt{2\pi/d}\right)^2}{2\sigma^2})\\
    &= \frac{1}{d}\theta_3\left(\pi\left(\frac{v-s_2}{d}\right),e^{-\frac{\pi f^2\sigma^2}{d}}\right).
\end{aligned}
\end{equation}
Note that this can be transformed exactly into the square case by only rescaling the value of $\sigma$. This gives the following achievable rate of
\begin{equation}
    I^{\mathrm{rect},f}_{d,\mathrm{analog}}(\sigma) = \frac{1}{2}(I^{\mathrm{sq}}_{d,\mathrm{analog}}(\sigma f) + I^{\mathrm{
    sq}}_{d,\mathrm{analog}}(\sigma/f)).
\end{equation}
We know that $I^{\mathrm{sq}}_{d,\mathrm{analog}}<\log_2(d)$ always and approaches $-\log_2(\sigma^2 e)$ as long as $d\sigma^2\gg1$. Hence if $d$ is large enough such that also $d\sigma^2\gg f^2$, then we have $I_{d,\mathrm{analog}}^{\mathrm{rect},f}(\sigma) \approx I_c(\sigma)$.\\
Assuming were are in a range of $\sigma$ such that $I^{\mathrm
sq}_{d,\mathrm{analog}}$ is concave in the range $[\sigma/f,\sigma f]$, we have
\begin{equation}
    I_{d,\mathrm{analog}}^{\mathrm{rect},f}(\sigma) \leq I_{d,\mathrm{analog}}^{\mathrm{sq}}\left(\frac{\sigma}{2f}(f^2 +1)\right) \leq I_{d,\mathrm{analog}}^{\mathrm{sq}}(\sigma),
\end{equation}
where the first inequality follows from concavity of $I_{d,\mathrm{analog}}^{\mathrm{sq}}(\sigma)$ and the second follows from $(f^2+1)\geq 2f$ hence giving $\sigma \leq \sigma\frac{f^2+1}{2f}$. Hence biasing does not improve the rate for square GKP for certain ranges of sigma. As shown in Fig. \ref{fig:10}, there is a range of values of $\sigma$ for each $d$ where $I_{d,\mathrm{analog}}^{\mathrm{sq}}$ is concave in $\sigma$ which notably is far from the limit of where it approaches coherent information (which is notably convex in $\sigma$). Over this range, rectangular biasing (provably) provides no benefit.

\subsection{General single-mode GKP with Gaussian displacement noise}
For any single-mode GKP, we would have the syndrome measurements given by $s_1$ and $s_2$. In general, the probability of bit (or phase) flip would simultaneously depend on the values of $s_1$ and $s_2$. Let us consider the event of logical error $X^uZ^v$ after correcting based on closest lattice point for syndrome values $s_1$ and $s_2$. This would follow some probability distribution
\begin{equation}
    p(u,v,s_1,s_2) \propto \sum_{\pmb{\xi}\in\mathcal{L}_{u,v,s_1,s_2}}\exp(-\frac{|\pmb{\xi}|^2}{2\sigma^2})
\end{equation}
where 
\begin{equation}
    \mathcal{L}_{u,v,s_1,s_2} = u\mathbf{x}_L + v\mathbf{z}_L + (MJ)^{-1}(\mathbf{s} - \mathbf{b}_0) + \mathcal{L}
\end{equation}
where $\mathbf{x}_L$ ($\mathbf{z}_L$) is the displacement corresponding to a logical $X$ ($Z$) operation. We also similarly would define $p(u,s_1,s_2) = \sum_v p(u,v,s_1,s_2)$. The polar codes associated to this would be obtained by looking at the two classical channels 
\begin{equation}
\begin{aligned}
     W_1((z+u,s_1,s_2)|z) &= p(u,s_1,s_2),\\ W_2((x+v,u,s_1,s_2)|x) &= p(u,v,s_1,s_2)
\end{aligned}
\end{equation}
which gives an achievable rate of
\begin{equation}
\begin{aligned}
    &R = \log_2(d)(I(W_1) + I(W_2) - 1)\\ &= \log_2(d) \\&\quad\quad+ \int ds_1ds_2\sum_{u,v} p(u,v,s_1,s_2)\log_2\left(\frac{p(u,v,s_1,s_2)}{p(s_1,s_2)}\right).
\end{aligned}
\end{equation}
Note that the main difference here arises from the fact that we need to consider the closest point from $\mathbf{b}_0$ which would depend on both $s_1$ and $s_2$ which creates decision boundaries corresponding to the Voronoi cell.

\section{Self orthogonal codes in $GF(d^2)$}\label{app:selfortho}
In this section we introduce stabilizer codes for qudits with a prime number of levels $d$, using self-orthogonal codes in $GF(d^2)$ using principles from \cite{959288}. We use the notation $[N,K]_d$ for a classical linear code which encodes $K$ dits into $N$ dits. Let us consider $N$ qudit Pauli operators written as
\begin{equation}
    X^{\pmb{a}^{(1)}}Z^{\pmb{a}^{(2)}} = \Pi_{i=1}^{N}X^{a_i^{(1)}}Z^{a^{(2)}_i},
\end{equation}
where $a_i^{(1)},a_i^{(2)} \in\mathbb{F}_d$. The condition for two Pauli operators to commute is given by
\begin{equation}
\begin{aligned}
[X^{\pmb{a}^{(1)}}Z^{\pmb{a}^{(2)}},X^{\pmb{b}^{(1)}}Z^{\pmb{b}^{(2)}}]=0\iff\\ \pmb{a}^{(1)}\cdot\pmb{b}^{(2)} = \pmb{a}^{(2)}\cdot\pmb{b}^{(1)} \mod d,
\end{aligned}
\end{equation}
which can be understood as having the symplectic inner product between the two $2N$ dimensional vectors $\pmb{a} = \begin{pmatrix}
    \pmb{a}^{(1)}\\\pmb{a}^{(2)}
\end{pmatrix}$ and $\pmb{b} = \begin{pmatrix}
    \pmb{b}^{(1)}\\\pmb{b}^{(2)}
\end{pmatrix}$ being zero. If we consider $[[N,K]]_d$ stabilizer codes that encode $K$ qudits in $N$ qudits, we have $N-K$ stabilizers which we will consider to be from the $N$ qudit Pauli group.

We consider some $[[N,K]]_d$ code defined by the stabilizer group  $\mathcal{S}\subseteq\mathcal{P}^{N}_d$ where $\mathcal{P}^{N}_d$ is the $N$ qudit Pauli group. Let us consider the set $C\subseteq \mathbb{F}^{2N}_d$ constructed from the stabilizer group $\mathcal{S}$ defined by
\begin{equation}
    C = \left\{\begin{pmatrix}
    \pmb{a}^{(1)}\\\pmb{a}^{(2)}
\end{pmatrix} \Big|\quad \hat{X}^{\pmb{a^{(1)}}}\hat{Z}^{\pmb{a^{(2)}}}\in\mathcal{S}\right\}.
\end{equation}
The set $C$ is linearly additive since if $\pmb{a},\pmb{b}\in C$ then so is $\alpha\pmb{a}+\beta\pmb{b}\in C$ for any $\alpha,\beta\in\mathbb{F}_d$. This means that $C$ can be considered as a linear code which is $[2N,N-K]_d$ since it consists of $N-K$ linearly independent generators. We can define an inner product over $\mathbb{F}^{2N}_d$ by $\langle\pmb{a},\pmb{b}\rangle = \pmb{a}^{(1)}\cdot\pmb{b}^{(2)} - \pmb{a}^{(2)}\cdot\pmb{b}^{(1)}$ which is a symplectic inner product. The relevance of this is that the two Pauli operators represented by $\pmb{a},\pmb{b}\in\mathbb{F}^{2N}_d$ would commute iff $\langle\pmb{a},\pmb{b}\rangle = 0$. Based on this we can define a symplectic dual code of $C$ as
\begin{equation}
    C^{\perp}=\{\pmb{a}|\forall\pmb{b}\in C, \langle\pmb{a},\pmb{b}\rangle = 0\},
\end{equation}
which would contain representations for all the Pauli operators which would give a logical error for the stabilizer code. For this to be a valid stabilizer code, we require that $C\subseteq C^{\perp}$ which is the condition of self orthogonality. This ensures that all the stabilizers commute with each other.

We can now see that elements of $C$ can also be considered to lie in the field of $\mathbb{F}^{N}_{d^2}$ since the Galois field $GF(d^2)$ can be defined as a field extension of $GF(d)$ if $d$ is a prime power. We will be restricting our discussion to having $d$ be prime which lets us easily define $\mathbb{F}_d$ to be integers modulo $d$. However since Galois fields require to have multiplication and division be defined along with addition and subtraction, this doesn't work for anything but $d$ being prime. Hence one needs to extend the field of $GF(d)$ by expressing a polynomial using coefficients from $GF(d)$ which doesn't actually have a solution in $GF(d)$. The solutions of this polynomial will then be used as a basis to express elements in extensions of $GF(d)$. As an example if $d$ is prime and $d=3\mod4$, there cannot exist any element which satisfies $x^2=-1\mod d$ for any $x\in \mathbb{F}_d$. Hence we can express elements in $GF(d^2)$ as $a = a^{(1)} + \gamma a^{(2)}$ where $a^{(1)},a^{(2)}\in \mathbb{F}_d$ and $\gamma\in\mathbb{F}_{d^2}\backslash\mathbb{F}_d$ satisfying $\gamma^2=-1\mod d$.

The connection of classical self-orthogonal codes to stabilizer codes has been explored in various works, with the pivotal work exploring qubit stabilizer codes obtained from $GF(4)$ \cite{681315} and further works extending this for qudits using $GF(d^2)$ \cite{959288,782103,1715533}. In the work by Ashikhmin and Knill \cite{959288}, they show that for $\pmb{a},\pmb{b}\in\mathbb{F}_{d^2}^N$ with $\pmb{a}=\pmb{a}^{(1)}+\gamma\pmb{a}^{(2)}$ and $\pmb{b} = \pmb{b}^{(1)}+\gamma\pmb{b}^{(2)}$ (here $\pmb{a}^{(j)},\pmb{b}^{(j)}\in\mathbb{F}^N_d$ for $j=1,2$) the following holds true 
\begin{equation}
\begin{aligned}
    \pmb{a}\cdot\pmb{b}^d = \sum_i a_i b_i^d = 0\mod d\\
    \implies \sum_i (a^{(1)}_ib_i^{(2)}-a^{(2)}_ib_i^{(1)}) = 0\mod d.
\end{aligned}
\end{equation}
Hence they define the inner product $\pmb{a}\pmb{b} =  \pmb{a}\cdot\pmb{b}^d$ which is also referred to as the Hermitian inner product in \cite{1715533}. Using this a dual of a code in $\mathbb{F}^N_{d^2}$ can be defined. This inner product being zero provides a sufficient condition for two Pauli operators to commute. Note that this does not offer a necessary condition for commuting Pauli operators which can be trivially noted by seeing that in general for $\pmb{a}\in \mathbb{F}^{N}_{d^2}$, $\pmb{a}\pmb{a}$ need not be zero. We note that in the case of $d=3\mod 4$, the self inner product takes on the nice form of
\begin{equation}
    \pmb{a}\pmb{a} = \sum_i((a^{(1)})^2 + (a^{(2)})^2) \mod d,
\end{equation}
which we will be making use of for proving our result in relation to achieving the capacity of pure-loss.

If we have a code $D\subseteq\mathbb{F}^{N}_{d^2}$ which is $[N,(N-K)/2]_{d^2}$ and the dual of it defined by
\begin{equation}
    D^\perp = \left\{\pmb{a}\in\mathbb{F}^{N}_{d^2}|\forall{\pmb{b}}\in D,\pmb{a}\pmb{b} = 0\right\},
\end{equation}
which satisfies the self-orthogonality condition $D\subseteq D^{\perp}$, this implies the existence of a quantum code which is $[[N,K]]_d$. Every element in $D$ satisfies self-orthogonality $\pmb{a}\pmb{a}=0$. The number of non-zero elements in $\mathbb{F}^N_{d^2}$ which satisfy this property is given by
\begin{equation}
    N_{\mathrm{self}} = \frac{1}{d}(d^{2N}+(d-1)(-d)^N)-1,
\end{equation}
as shown in \cite{1362919}. We now restate Lemma 7 from \cite{959288} which will be crucial in our analysis.
\begin{lemma}
    Consider the set $\mathcal{T}$ consisting of all possible $[N,(N-K)/2]_{d^2}$ self orthogonal codes. The number of codes in $\mathcal{T}$ which contains a given non-zero self-orthogonal vector $\pmb{a}$ is independent of $\pmb{a}$. Hence all non-zero self-orthogonal vectors appear the same number of codes in $\mathcal{T}$.
\end{lemma}
\begin{corollary}\label{cor:fbalanced}
For any function $f:\mathbb{F}_{d^2}^N\to\mathbb{R}$
    \begin{equation}
        \frac{1}{|\mathcal{T}|}\sum_{D\in\mathcal{T}}\sum_{\pmb{a}\in D}f(\pmb{a}) = f(0)+ \frac{d^{N-K}-1}{N_{\mathrm{self}}}\sum_{\pmb{a},\pmb{aa}=0,\pmb{a}\neq 0}f(\pmb{a}).
    \end{equation}
\end{corollary}
This follows from the basic-averaging lemma \cite{641543}. Note importantly that the sum is restricted to looking exactly at self orthogonal $a$. If we assume that $d$ is a prime number such that $d = 3\mod 4$, we can assume a basis using $\gamma^2+1=0$ and so $\gamma=i$ and $\gamma_0=0$. This means that the self orthogonal $\pmb{a}=\pmb{a}^{(1)}+i\pmb{a}^{(2)}$ will satisfy $|\pmb{a}^{(1)}|^2 + |\pmb{a}^{(2)}|^2 = 0\mod d$.

\section{Achieving the capacity for pure-loss}\label{app:loss}
\begin{lemma}
    The infidelity of an infinite-energy GKP code experiencing pure loss (transmittance $\eta$) followed by the transpose recovery with an underlying symplectically integral lattice $\mathcal{L}$ encoding a finite dimension, can be upper bounded as
    \begin{equation}
        \epsilon \leq \frac{1}{4}\sum_{\pmb{x}\in\mathcal{L}^\perp\backslash{0}}e^{-\frac{\eta}{1-\eta}|\pmb{x}|^2}.
    \end{equation}
\end{lemma}
The above lemma is taken from \cite{zheng2024performanceachievableratesgottesmankitaevpreskill}. We now proceed to prove the main result.
\begin{theorem}
    There exists a sequence of qudit stabilizer codes $[[N,K]]_d$ for $d$ being prime with $d=3\mod 4$, with increasing $N,d$ such that 
    \begin{equation}
        \log_2(d)\frac{K}{N} = \log_2\left(\frac{\eta}{1-\eta}\right) - \tilde{\epsilon}
    \end{equation}
    and $\tilde{\epsilon}$ can be made arbitrarily small simultaneously while the infidelity of this sequence of codes using transpose recovery after pure loss of transmittance $\eta$ converges to zero as $N,d\to \infty$ with $d\ln(d)\ll N\ll e^{\pi d\frac{\eta}{1-\eta}}$.
\end{theorem}
\begin{proof}
    We first begin by noting a property of the summation ($g>0$)
    \begin{equation}
        S_g = \sum_{z\in\mathbb{Z}^{2N}:|z|^2=0\mod d} e^{-\pi g|z|^2/d},
    \end{equation}
    which is equivalent to
    \begin{equation}
        S_g = \frac{1}{d}\theta_3(e^{-\pi g/d})^{2N} + \frac{2}{d}\sum_{j=1}^{\lfloor (d-1)/2\rfloor }\theta_3(e^{-\pi g/d}\omega^j)^{2N},\label{eq:Sgsum}
    \end{equation}
    where $\omega = \exp(2\pi i/d)$. The use of $\omega$ ensures that any terms that have $|z|^2\neq 0\mod d$ will cancel out. Now using Eq. \eqref{eq:dualtheta} we note that 
    \begin{equation}
    \begin{aligned}
        \theta_3(e^{-\pi g/d}\omega^j) = \theta_3(e^{-\pi \frac{d}{g-2ij}})\sqrt{\frac{d}{g-2ij}}.
    \end{aligned}
    \end{equation}
    For any lattice, it is trivial that $|\Theta_{\mathcal{L}}(re^{i\theta})|\leq \Theta_\mathcal{L}(r)$ for $r\in\mathbb{R^+}$, $\theta\in[-\pi,\pi]$. Hence we get that
    \begin{equation}
        |\theta_3(e^{-\pi g/d}\omega^j)| \leq \sqrt{\frac{d}{g}}\left(\left(\frac{g^2}{g^2+4j^2}\right)^{1/4}\theta_3(e^{-\pi \frac{d g}{g^2+4j^2}})\right).
    \end{equation}
    We will now note that in the summation in Eq. \ref{eq:Sgsum}, the central term of $j=0$ dominates by an exponential amount compared to any $j\neq 0$. To check this, we first consider the following function
    \begin{equation}
        h(t) = t^{1/4}\theta_3(e^{-\pi t}) = t^{-1/4}\theta_3(e^{-\pi/t}) = h(1/t).
    \end{equation}
    It clearly follows that $h(t)$ has a local minima at $t=1$ since its derivative equals zero, and perturbing around $t=1$ only increases this function. To see why it is increasing for $t>1$, note that 
    \begin{equation}
    \begin{aligned}
        h'(t) &= \frac{t^{-3/4}}{4} - \pi t^{1/4}e^{-\pi t}\theta'_3(e^{-\pi t})\\
        &= \frac{t^{-3/4}}{4}(1 - 4\pi te^{-\pi t}\theta_3'(e^{-\pi t})),
    \end{aligned}
    \end{equation}
    which we know is zero at $t=1$. Also
    \begin{equation}
        \theta'_3(e^{-\pi t}) = 2\sum_{n=1}^{\infty}n^2e^{-(n^2-1)\pi t},
    \end{equation}
    which shows that $t>1\implies \theta'_3(e^{-\pi t})<\theta'_3(e^{-\pi}) = e^{\pi}/4\pi$ which gives $h'(t)\geq t^{-3/4}(1-te^{1-\pi t})$. Trivially, $ te^{-\pi t}<1$ for all $t>1$. Hence it follows that for $t>1$, $h'(t)>0$. 
    
    We substitute
    \begin{equation}
        t= \max\left\{\frac{g^2+4j^2}{dg},\frac{dg}{g^2+4j^2}\right\},
    \end{equation}
    which ensures $t>1$ and since $h(t)=h(1/t)$ this gives
    \begin{equation}
    \begin{aligned}
        \frac{|\theta_3(e^{-\pi g/d}\omega^j)|}{\theta_3(e^{-\pi g/d})} \leq \left(\frac{t^{1/4}\theta_3(e^{-\pi t})}{\left(\frac{d}{g}\right)^{1/4}\theta_3(e^{-\frac{\pi d}{g}})}\right).
    \end{aligned}
    \end{equation}
    Let us suppose $j$ always lies in $1\leq 4j^2\leq d^2-d-g^2$. This already ensures that $t<d/g$ which means the above expression is strictly less than 1. In this range, the maximum value of $t^{1/4}\theta_3(e^{-\pi t})$ occurs at $4j^2$ being closest to $d^2-d-g^2$ which gives 
    \begin{equation}
    \begin{aligned}
        \frac{|\theta_3(e^{-\pi g/d}\omega^j)|}{\theta_3(e^{-\pi g/d})} &\leq \left(\left(1-\frac{1}{d}\right)^{1/4}\frac{\theta_3(e^{-\frac{\pi d}{g}(1-d^{-1})})}{\theta_3(e^{-\frac{\pi d }{g}})}\right)\\
        &\leq \left(1-\frac{0.9}{d}\right)^{1/4},\label{eq:theta3jineq}
    \end{aligned}
    \end{equation}
    when $e^{\pi d/g}\gg d$ since $\frac{\theta_3(e^{-\frac{\pi d}{g}(1-d^{-1})})}{\theta_3(e^{-\frac{\pi d }{g}})} \leq 1+\mathcal{O}(e^{-\pi d/g})$. This range already will contain all possible $j$ since it only goes up till $(d-1)/2$. We can restrict to this choice of $j$ without loss of generality since $\omega^j = \omega^{j \mod d}$ where the modulo restricts it to the range $-\frac{d-1}{2}$ to $\frac{d-1}{2}$. We now fix $g=\frac{1-\eta}{\eta}$ and using Eqs. \eqref{eq:dualtheta} and \eqref{eq:theta3jineq}, we obtain the following
    \begin{equation}
    \begin{aligned}
        S_{\frac{1-\eta}{\eta}} \leq \frac{1}{d}&d^N\left(\frac{\eta}{1-\eta}\right)^N\theta_3(e^{-\pi d(\frac{1-\eta}{\eta})})^{2N}\\
        &\times\left(1+d\left(1-\frac{0.9}{d}\right)^{\frac{N}{2}}\right).
    \end{aligned}\label{eq:Sgineq}
    \end{equation}

    Now assuming we use a transpose recovery on a $[[N,K]]_d$ concatenated square GKP code, the infidelity can be bounded by
    \begin{equation}
        4\epsilon \leq \sum_{\pmb{x}\in\mathcal{L}^\perp\backslash{0}}e^{-\pi\frac{\eta}{1-\eta}|x|^2},
    \end{equation}
    which using Eq. \eqref{eq:dualtheta} can be equivalently written as
    \begin{equation}
        4\epsilon \leq d^K\left(\frac{1-\eta}{\eta}\right)^N\Theta_\mathcal{L}(e^{-\pi\frac{1-\eta}{\eta}}) - 1,
    \end{equation}
    since the $\det(\mathcal{L})=d^K$ which is the total logical dimension.
    
    Let us consider the set $\mathcal{T}$ composed of self orthogonal codes $[N,(N-K)/2]_{d^2}$ which can be equivalently mapped to a set of stabilizer lattices $\mathcal{T}_\mathcal{L}$ composed of taking self-orthogonal codes $D$ in $\mathcal{T}$ and mapping those to $2N$ dimensional lattices. The explicit mapping is defined by
    \begin{equation}
    \begin{aligned}
        \mathcal{L}_D = \Bigg\{&\frac{\pmb{z}}{\sqrt{d}}|\pmb{z}=(\pmb{z}^{(1)},\pmb{z}^{(2)})\in\mathbb{Z}^{2N}, \\&\text{ with }((\pmb{z}^{(1)}+i\pmb{z}^{(2)})\mod d)\in D\Bigg\},
    \end{aligned}
    \end{equation}
    which will represent a valid stabilizer lattice due to the self-orthogonality condition of $D\subseteq D^\perp$. Hence $\mathcal{T}_\mathcal{L}$ is composed of all the lattices $\mathcal{L}_D$ obtained from $D\in\mathcal{T}$. Due to the balanced nature of $\mathcal{T}$ over the set of self-orthogonal elements in $F^{N}_{d^2}$ (see corollary \ref{cor:fbalanced}), we get
    \begin{equation}
    \begin{aligned}
        \frac{1}{|\mathcal{T}_\mathcal{L}|}&\sum_{\mathcal{L}\in\mathcal{T}_{\mathcal{L}}}\Theta_{\mathcal{L}}(e^{-\pi\frac{1-\eta}{\eta}})\\&< \sum_{\pmb{z}\in\mathbb{Z}^{2N}:\pmb{z}=0\mod d}e^{-\pi|\pmb{z}|^2\frac{1-\eta}{\eta}} \\&\quad+ \frac{d^{N-K}-1}{N_{\mathrm{self}}}\sum_{\pmb{z}\in\mathbb{Z}^{2N}:|\pmb{z}|^2=0\mod d}e^{-\pi|\pmb{z}|^2\frac{1-\eta}{\eta}}\\ 
        &< \theta_3(e^{-\pi d\frac{1-\eta}{\eta}})^{2N} + \frac{d^{N-K}-1}{N_{\mathrm{self}}}S_{\frac{1-\eta}{\eta}},
    \end{aligned}
    \end{equation}
    where the first term comes from counting over all lattice points in $\sqrt{d}\mathbb{Z}^{2N}$ that are equivalent to the $0$ codeword and the second term comes from counting all the points that are equivalent to some self-orthogonal codeword since we are using the basis $\{1,i\}$ for $F_{d^2}$ (a valid choice since $d= 3\mod 4$). By representing $\pmb{z}\mod d=(\pmb{a}^{(1)},\pmb{a}^{(2)})$ where the self orthogonality for $\pmb{a} = \pmb{a}^{(1)}+i\pmb{a}^{(2)}$ is equivalent to having $|\pmb{z}|^2 = 0\mod d$. There is over counting of codewords equivalent to $0$ in the RHS, hence giving the inequality Since we are evaluating an average over the whole set $\mathcal{T}_{\mathcal{L}}$, this means that there must exist some lattice $\mathcal{L}_D\subseteq\mathcal{T}_\mathcal{L}$ which satisfies
    \begin{equation}
        \Theta_{\mathcal{L}_D}(e^{-\pi\frac{1-\eta}{\eta}})< \theta_3(e^{-\pi d\frac{1-\eta}{\eta}})^{2N} + \frac{d^{N-K}-1}{N_{\mathrm{self}}}S_{\frac{1-\eta}{\eta}}.
    \end{equation}
    Note that for any $N$ and large enough $d$,
    \begin{equation}
        \frac{d^{N-K}-1}{N_{\mathrm{self}}}< d^{-N-K+1}(1 + 1.1d^{-N+1}),
    \end{equation}
    which when combined with equation Eq. \eqref{eq:Sgineq} gives the infidelity for $\mathcal{L}_D$ to satisfy
    \begin{equation}
        \begin{aligned}
            4\epsilon \leq &\left(d^{K/N}\theta_3(e^{-\pi d\frac{1-\eta}{\eta}})^2\left(\frac{1-\eta}{\eta}\right)\right)^N\\
            &+\left(1+\frac{1.1}{d^{N-1}}\right)\left(1+d\left(1-\frac{0.9}{d}\right)^{\frac{N}{2}}\right)\theta_3(e^{-\pi d\frac{\eta}{1-\eta}})^{2N}\\
            &- 1.
        \end{aligned}
    \end{equation}
    Note that if the lattice $\mathcal{L}_D$ is associated to a GKP code with the rate $R$ in per qubit, then we have $2^R = d^{K/N}$. For a certain rate to be achievable, we require that there must be a sequence of codes with certain encoding and decoding such that the infidelity in the limit of $N\to\infty$ is equal to zero. 
    
    Let us assume that we have the following guarantees in relating $N$ and $d$ which are
    \begin{align}
        N&\gg d\ln(d),\\
        N&\ll e^{\pi d\frac{\eta}{1-\eta}}.
    \end{align}
    Assuming that $d$ is then large enough, we can always find a large enough $d$ such that
    \begin{align}
        \theta_3(e^{-\pi d\frac{1-\eta}{\eta}})^{-2} &\geq 1- 4.1e^{-\pi d\frac{1-\eta}{\eta}},\\
        (1+1.1d^{-N+1})&\left(1+d\left(1-0.9d^{-1}\right)^{\frac{N}{2}}\right)\theta_3(e^{-\pi d\frac{\eta}{1-\eta}})^{2N}\nonumber\\&\leq 1+4.1Ne^{-\pi d\frac{\eta}{1-\eta}},
    \end{align}
    following which we force
    \begin{equation}
        d^{K/N}\leq \left(\frac{\eta}{1-\eta}\right)(1- 4.1e^{-\pi d\frac{1-\eta}{\eta}})(1-N^{-1+\delta}),\label{eq:dKNcap}
    \end{equation}
    where $\delta>0$. The residual factor of $(1-N^{-1+\delta})$ is to ensure that raising to the power of $N$ still makes the upper bound to the infidelity approach zero. For this example we can choose $\delta = 1/2$. This gives
    \begin{equation}
        \begin{aligned}
            4\epsilon \leq &\left(d^{K/N}\theta_3(e^{-\pi d\frac{1-\eta}{\eta}})^2\left(\frac{1-\eta}{\eta}\right)\right)^N\\
            &\quad+4.1Ne^{-\pi d\frac{\eta}{1-\eta}}\\
            \leq &1.1e^{-N^{\delta}}+ 4.1Ne^{-\pi d\frac{\eta}{1-\eta}},
        \end{aligned}
    \end{equation}
    which using the relations above can be made arbitrarily small by appropriate choice of $N$ and $d$. Specifically, we can define $N = \lfloor e^{d\frac{\eta}{1-\eta}}\rfloor$ for any given $d$. Then we can keep increasing $d$ and $N$ whilst defining $K$ such that it satisfies Eq. \eqref{eq:dKNcap} and the value of $\epsilon$ for this sequence of codes will be decreasing and tend to zero with the achievable rate being
    \begin{equation}
    \begin{aligned}
        R_d &= \log_2(d^{K/N}) \\
        &= \log_2\left(\frac{\eta}{1-\eta}\right) - \mathcal{O}(e^{-\pi d\frac{1-\eta}{\eta}})- \mathcal{O}(N^{-1+\delta}),
    \end{aligned}
    \end{equation}
    which then clearly shows that this sequence of codes achieves the rate $\log_2(\eta/(1-\eta))$ which is the capacity of the loss channel.
\end{proof}

\end{document}